\documentclass{amsart}
\usepackage{amssymb,amsmath}
\usepackage{mathrsfs}
\usepackage{enumerate}
\usepackage[colorlinks=true,citecolor=blue,linkcolor=blue]{hyperref}

\newcommand{\brs}[1]{\left(#1\right)}
\newcommand{\set}[1]{\left\{#1\right\}}
\newcommand{\SqBrs}[1]{\left[#1\right]}
\newcommand{\anglb}[1]{\langle#1\rangle}
\newcommand{\abs}[1]{\left|#1\right|}

\numberwithin{equation}{section}

\newtheorem{thm}{Theorem}[section]
\newtheorem{example}[thm]{Example}

\newtheorem{prop}[thm]{Proposition}
\newtheorem{cor}[thm]{Corollary}
\newtheorem{defn}[thm]{Definition}
\newtheorem{lemma}[thm]{Lemma}

\theoremstyle{remark}
\newtheorem{remark}[thm]{Remark}

\begin{document}
\title{A Dixmier trace formula for the density of states}

\author{N.\,Azamov, E.\,McDonald, F.\,Sukochev, D.\,Zanin}
\address{University of New South Wales, Kensington, NSW, 2052, Australia}
\email{nurulla.azamov@unsw.edu.au}
\email{edward.mcdonald@unsw.edu.au}
\email{f.sukochev@unsw.edu.au}
\email{d.zanin@unsw.edu.au}

\begin{abstract} 
A version of Connes trace formula allows to associate a measure on the essential spectrum of a Schr\"odinger operator
with bounded potential. In solid state physics there is another celebrated measure associated with such operators --- the density of states.
In this paper we demonstrate that these two measures coincide.  We show how this equality can be used to give explicit formulae for the density of states in some circumstances.
\end{abstract}
\maketitle

\tableofcontents

\section{Introduction}

Let $d \geq 2$, and let
\begin{equation} \label{F: H}
H = -\Delta + V
\end{equation}
be a Schr\"odinger operator on ${\mathbb{R}}^d$, where $\Delta = \sum_{j=1}^d \frac{\partial^2}{\partial x_j^2}$ is the Laplace operator and~$V$ is a bounded real-valued measurable potential $V \in L_\infty({\mathbb{R}}^d).$
The \emph{density of states} (or \emph{DOS}) is a Borel measure $\nu_H$ on ${\mathbb{R}}$ naturally associated to $H,$ see e.g. \cite{AiWa,PasFig,CarmonaLacroix,Lang_genetic},
defined as follows.
Let $L > 0$, and let $H_L$ be the restriction of $H$ to the cube $(-L,L)^d$ with Dirichlet boundary conditions (for a definition see e.g. \cite[\S VI.4.4]{KatoBook}, \cite[\S XIII.15]{ReedSimonIV}). 

Let $I \subset {\mathbb{R}}$ be a bounded interval, and let $N_L(I)$ be the number of eigenvalues with multiplicities of~$H_L$ in~$I$ (which is necessarily finite, since $H_L$ has compact resolvent). The density of states measure $\nu_H$ of $I$
is defined as
\begin{equation} \label{def_of_DOS_via_cubes}
\nu_H(I) = \lim_{L \to \infty} \frac {N_L(I)}{{\mathrm{Vol}}((-L,L)^d)}.
\end{equation}

The DOS measure does not always exist, see e.g. \cite[p.\,513]{SimonSemigroups}. However, it is well known to exist for certain classes of Hamiltonians important for solid state physics such as those corresponding to periodic, almost periodic and ergodic potentials, see e.g. \cite{AiWa,Beresin_Shubin,PasFig,ReedSimonIV,SimonSemigroups}. 
Another point to mention here is that the density of states measure~$\nu_H$ has several definitions. The difference is in the choice of domain in the limit \eqref{def_of_DOS_via_cubes}: one can replace the cubes $\{(-L,L)^d\}_{L > 0}$ with a family of balls or other domains. There is also some variation in the choice of boundary conditions used to define $H_L$ (such as Dirichlet, Neumann, periodic, etc.).
For our purposes it will be convenient to use yet another definition 
(see e.g. \cite[(C41)]{SimonSemigroups}, \cite[(1.2)]{DIT2001}) 
in terms of the spectral projections of $H,$ as follows
\begin{equation} \label{def_of_DOS}
\nu_H(I) = \lim_{R \to \infty} \frac {1}{{\mathrm{Vol}}(B(0,R))} {\rm Tr}\brs{M_{\chi_{B(0,R)}} \chi_I(H) M_{\chi_{B(0,R)}}},
\end{equation}
where $B(0,R)$ is the ball of radius $R$ centred at zero, $M_f$ is the operator of multiplication by a function $f$ on $L_2({\mathbb{R}}^d),$ ${\rm Tr}$ is the standard operator trace
and $\chi_A$ is the indicator function of a set $A.$
(We use notation $M_f$ when it is necessary to distinguish a function $f$ from the operator of multiplication by $f,$ however occasionally when there is no danger of confusion we write $f$
meaning $M_f,$ especially in the Schr\"odinger operator $-\Delta + V$).
It is known that, assuming existence, these different definitions of DOS coincide at least for such important classes of potentials as periodic or ergodic, the latter including random and almost periodic potentials, see e.g. \cite[Theorem C.7.4]{SimonSemigroups}
and \cite{DIT2001}.
In this paper we will assume existence of the limit \eqref{def_of_DOS}.

 The density of states has attracted substantial attention in the mathematical literature. Items of particular interest have been the existence of $\nu_H$ in various circumstances,
the asymptotic behaviour of the integrated density of states function $\nu_H((-\infty,\lambda])$ as $\lambda\to\infty$ and its continuity properties \cite{Shubin-UMN, BourgainKlein}. As an example
of these results, we mention in particular that it follows from the work of Shubin \cite[Theorem 4.5]{Shubin-UMN} that if $V$ is smooth and almost periodic (see \cite[\S 1.2]{Shubin-UMN} for details) then $\nu_H$
exists and
$$\nu_H((-\infty,\lambda]) = \frac{\omega_d}{(2\pi)^d}\lambda^{\frac{d}{2}}+O(\lambda^{\frac{d}{2}-1})$$
as $\lambda\to\infty$, where
\begin{equation*}
\omega_d := \frac{2\pi^{d/2}}{\Gamma\left(\frac{d}{2}\right)}
\end{equation*}
is the $(d-1)$-volume of the unit sphere $S^{d-1}$. 
More recently, Bourgain and Klein \cite[Theorem 1.1]{BourgainKlein} proved among other results that in dimensions $d = 1,2,3$ the density of states (defined in terms of cubes, as in \eqref{def_of_DOS_via_cubes}) satisfies the following local $\log$-H\"older continuity property:
$$\nu_H([E,E+\varepsilon]) \leq C_{d,V,E}\log(\varepsilon^{-1})^{-2^{-d}},\quad E \in {\mathbb{R}}, \varepsilon \leq \frac{1}{2}$$
whenever $V$ is bounded and $\nu_H$ exists.

\smallskip
The second measure which can be associated with $H$ comes from a version of Connes' trace formula \cite{action}, \cite[Corollary 7.21]{GVF}, \cite[Theorem 11.7.5]{LSZ2012}.
One form of Connes trace formula asserts that for all continuous and compactly supported functions $f$ on ${\mathbb{R}}^d$, we have:
\begin{equation} \label{F: Connes' trace formula}
{\rm Tr}_{\omega}\brs{M_f (1-\Delta)^{-d/2}} = \frac{{\omega_d}}{d(2\pi)^d}\int_{{\mathbb{R}}^d} f(t)\,dt,
\end{equation}
where  
${\rm Tr}_{\omega}$ is a Dixmier trace on the ideal ${\mathcal{L}}_{1,\infty}(L_2({\mathbb{R}}^d)).$
For our purpose it is desirable to rewrite this formula
in the Fourier transform picture, as follows:
\begin{equation}\label{momentum_cif}
{\rm Tr}_{\omega}\brs{f(-i\nabla) M^{-d}_{\anglb{x}} } = \frac{{\omega_d}}{d(2\pi)^d} \int_{{\mathbb{R}}^d} f(t)\,dt,
\end{equation}
where $\nabla = (\partial_1,\ldots,\partial_d)$ is the gradient operator, $f(-i\nabla)$ is defined via functional calculus, and 
$$
\anglb{x} = (1+\abs{x}^2)^{1/2}.
$$
We would like to rewrite this formula in terms of the Laplacian $-\Delta$ rather than the gradient operator $\nabla$. To this end, consider the case 
where $f$ is a radial function, and therefore $f(-i\nabla)$ can be written as $g(-\Delta)$ for some continuous compactly supported function $g$ on $[0,\infty)$. Then by switching to polar coordinates we have

\begin{align*}
{\rm Tr}_{\omega}\brs{g(-\Delta) M^{-d}_{\anglb{x}} } &= \frac{\omega_d}{d(2\pi)^d}\int_{{\mathbb{R}}^d} g(|x|^2)\,dx\\
                                                    &= \frac{\omega_d}{d(2\pi)^d} \int_{S^{d-1}}\int_0^\infty g(r^2)r^{d-1}\,drds\\
                                                    &= \frac{\omega_d^2}{d(2\pi)^d} \int_0^\infty g(r^2)r^{d-1}\,dr\\
                                                    &= \frac{\omega_d^2}{2d(2\pi)^d}\int_0^\infty g(\lambda)\lambda^{\frac{d}{2}-1}\,d\lambda.
\end{align*}

Now, let $V \in L_\infty({\mathbb{R}}^d)$ be a real-valued potential. Our main result is the following:

\begin{thm} \label{T: main} Let $H = -\Delta+M_V$ be a Schr\"odinger operator on $L_2({\mathbb{R}}^d)$, where $V$ is a bounded real-valued measurable potential.
For any $g \in C_c({\mathbb{R}})$ the operator
$$
g(H) M^{-d}_{\anglb{x}} 
$$
belongs to the weak trace-class ideal ${\mathcal{L}}_{1,\infty}(L_2({\mathbb{R}}^d))$. If we assume that the density of states of $H$ (defined according to \eqref{def_of_DOS}) exists and is a Borel measure $\nu_H$ on $\mathbb{R}$, then for every Dixmier trace
${\rm Tr}_\omega$ on~${\mathcal{L}}_{1,\infty}$ there holds the equality
\begin{equation} \label{F: phi(..)=mu(g)}
{\rm Tr}_\omega\brs{g(H) M^{-d}_{\anglb{x}} } = \frac{{\omega_d}}{d}  \int_{{\mathbb{R}}} g\,d\nu_H.
\end{equation}

\end{thm}

It is instructive to consider the simplest case, $V = 0$, which also serves to compute the constant in \eqref{F: phi(..)=mu(g)}. 
\begin{example} \label{E: 1} For $H_0 = -\Delta$ we have 
\begin{equation*}
{\rm Tr}_{\omega}(g(H_0)M^{-d}_{\anglb{x}}) = \frac{{\omega_d}}{d}\int_{\mathbb{R}} g(\lambda) d\nu_{H_0}(\lambda)
\end{equation*}
for all $g \in C_c(\mathbb{R})$.
\end{example}
\begin{proof}
We shall verify that:
\begin{equation*}
    {\rm Tr}_\omega(e^{-tH_0}M^{-d}_{\anglb{x}}) = \frac{{\omega_d}}{d}\int_0^\infty e^{-t\lambda}\,d\nu_{H_0}(\lambda),\quad t > 0.
\end{equation*}
This suffices to ensure that the equality holds for all $g \in C_c(\mathbb{R})$ (see the argument in Remark~\ref{laplace_uniqueness}) provided both sides exist.
The existence of the left-hand side follows from classical Cwikel estimates, or may be derived from the Cwikel estimates given below in Section~\ref{cwikel_section}.
That there is indeed an explicitly computable density of states measure for this case is well known (see e.g. \cite[Theorem C.7.7]{SimonSemigroups}).

Connes' trace formula in the form \eqref{momentum_cif} yields:
\begin{align*}
    {\rm Tr}_{\omega}(e^{t\Delta}M^{-d}_{\anglb{x}}) &= \frac{\omega_d}{d(2\pi)^d}\int_{\mathbb{R}^d} e^{-t|x|^2}\,dx\\
                                                    &= \frac{\omega_d}{d(2\pi)^d}\left(\frac{\pi}{t}\right)^{d/2}\\
                                                    &= \frac{\omega_d}{d}(4\pi t)^{-d/2}.
\end{align*}
We may compare this to $\int_0^\infty e^{-t\lambda} d\nu_{H_0}(\lambda)$ using \eqref{def_of_DOS}. According to \cite[Proposition C.7.2]{SimonSemigroups}, it suffices to compute:
\begin{equation*}
    \lim_{R\to\infty}\frac{1}{|B(0,R)|}{\rm Tr}(M_{\chi_{B(0,R)}}e^{-tH_0}).
\end{equation*}
The semigroup $e^{-tH_0}$ has integral kernel (see e.g. \cite[\S IV.7, Example 3]{ReedSimonII}):
\begin{equation*}
    K_t(x,y) = (4\pi t)^{-d/2}e^{-\frac{|x-y|^2}{4t}}.
\end{equation*}
Hence $K_t(x,x) = (4\pi t)^{-d/2}$ is constant, and we have:
\begin{equation*}
    \frac{1}{|B(0,R)|}{\rm Tr}(M_{\chi_{B(0,R)}}e^{-tH_0}) = \frac{1}{|B(0,R)|}\int_{B(0,R)} (4\pi t)^{-d/2}\,dx = (4\pi t)^{-d/2}.
\end{equation*}
\end{proof} 

Theorem~\ref{T: main} observes a direct connection between two measures which can naturally be associated with the operator \eqref{F: H}. 
Since these two measures {\it a priori} have very different definitions, this connection ought to be considered as somewhat surprising. 
At the same time, both measures do share some obvious common properties.
Indeed, both measures are invariants of a Schr\"odinger operator (\ref{F: H}),
both are suppported on the essential spectrum of $H$ and they both exhibit certain robustness. Namely, the Dixmier trace ${\rm Tr}_\omega,$
used in the definition of one of these measures, is insensitive to trace class perturbations, while the density of states measure $\nu_H$ is insensitive to localised perturbations $V_0+V$ 
of the bounded potential~$V_0,$ \cite[Theorems C.7.7 and C.7.8]{SimonSemigroups} reflecting the fact that DOS is a property of the behaviour of the potential~$V$ at infinity.

As mentioned above, throughout this paper we assume existence of the density of states. That is, we assume existence of the limit \eqref{def_of_DOS}. It is noteworthy however that the left hand
side of \eqref{F: phi(..)=mu(g)} is meaningful for an arbitrary bounded potential $V$ and defines a Borel measure on the spectrum of $H=-\Delta+M_V$. In the event that the density of states does indeed exist, \eqref{F: phi(..)=mu(g)} implies
that the value of the Dixmier trace ${\rm Tr}_\omega(f(H)M_{\anglb{x}}^{-d})$ is independent of the extended limit $\omega$. In general, an operator $T$ in the weak trace ideal is called Dixmier measurable
if ${\rm Tr}_\omega(T)$ is independent of the choice of Dixmier trace ${\rm Tr}_\omega$. Hence, existence of the density of states implies that $f(H)M_{\anglb{x}}^{-d}$ is Dixmier measurable for all $f\in C_c({\mathbb{R}})$. We conjecture
that the converse holds.

\section{Example computations}
    Before giving the full proof of Theorem \ref{T: main}, we can demonstrate its utility to explicitly compute the density of states in a number of examples. To the best of our knowledge,
    these formulae are new. In these examples we have not attempted to prove that the density of states exists, nonetheless using \eqref{F: phi(..)=mu(g)} we can compute it conditional on the assumption
    of existence.
    
    Given a potential $V$, we denote the density of states measure of $H = -\Delta+V$ by $\nu_{H}$ (assuming that the measure exists). We will also write $H = -\Delta+M_V$
    whenever there is potential for confusion between the pointwise multiplier $M_V$ and the function $V$.
    
\subsection{Radially homogeneous potentials}
    In this subsection we consider potentials $V$ which are positively homogeneous in the sense that $V(tx) = V(x)$ for all $t > 0$ and all $x \in {\mathbb{R}}^d$.
    
    \begin{thm}\label{homogeneous_potential_theorem}
        Denote $H_0 = -\Delta$ and let $H = H_0 + M_V,$ where $V \in L_\infty({\mathbb{R}}^d)$ is such that $V(tx) = V(x)$ for all $t > 0$ and all $x \in {\mathbb{R}}^d$.
        We assume that the density of states of $H$ exists. 
        Then the density of states of $H$ is the average over $\xi \in S^{d-1}$ of the density of states of $H_0+V(\xi),$
        in other words:
        \begin{equation*}
            \nu_{H_0+V} = \frac{1}{\omega_d} \int_{S^{d-1}} \nu_{H_0+V(\xi)} \,d\xi.
        \end{equation*}
    \end{thm}
    The spectral and scattering theory of Schr\"odinger operators with positively homogeneous potentials have been studied by Herbst and Skibsted \cite{Herbst1991, HerbstSkibsted2001, HerbstSkibsted2004},
    although we are not aware of any results which imply the existence of the density of states for these potentials.

    The computation in this case is based on a version of Connes trace theorem proved by two of the authors in \cite{Dao1}, which is stated as follows (c.f. \cite[Theorem 1.2]{Dao1}).
    Let~$\Pi$ be the $C^*$-subalgebra of ${\mathcal{B}}(L_2({\mathbb{R}}^d)),$ generated by the algebras of pointwise multipliers
$$
   \set{M_f \colon f \in L_\infty({\mathbb{R}}^d)}
$$
and homogeneous Fourier multipliers
$$
   \set{ g\brs{\frac {D_1}{(-\Delta)^{1/2}}, \ldots, \frac {D_d}{(-\Delta)^{1/2}}}  \colon g \in L_\infty\brs{S^{d-1}}}, 
$$
where $D_k = \frac1 i \frac {\partial}{\partial x_k}$ and $S^{d-1}$ is the $d-1$ dimensional unit sphere. 
Theorem 1.2 of \cite{Dao1} asserts that there exists a unique $C^*$-algebra homomorphism
$$
   {\rm sym} \colon \Pi \to L_\infty\brs{{\mathbb{R}}^d \times S^{d-1}},
$$
such that ${\rm sym}\brs{M_f} = f \otimes 1$ and ${\rm sym} \brs{ g\brs{\frac {D_1}{\sqrt{-\Delta}}, \ldots, \frac {D_d}{\sqrt{-\Delta}}}} = 1 \otimes g.$
Theorem 1.5 of \cite{Dao1} asserts that for $T \in \Pi$ and $g\in L_\infty({\mathbb{R}}^d)$ is compactly supported, then the opreator $TM_g \anglb{\nabla}^{-d}$ belongs to ${\mathcal{L}}_{1,\infty}$ and 
\begin{equation} \label{F: Dao1 formula}
   {\rm Tr}_\omega\brs{TM_g \anglb{\nabla}^{-d} } = \int_{{\mathbb{R}}^d \times S^{d-1}} {\rm sym}\brs{T}g.
\end{equation}
This theorem is a version of Connes' trace theorem \cite[Theorem 1]{action}.

\begin{proof}[Proof of Theorem \ref{homogeneous_potential_theorem}]
    Let $s > 0.$ For each $N > 0$, performing $N$-fold iteration of Duhamel's formula gives the Dyson expansion with remainder term:
    \begin{align*}
        e^{-sH} &=  e^{-sH_0} + \sum_{n=1}^N \frac{(-s)^n}{n!}\int_{\Delta_n} e^{-\theta_0 sH_0}M_Ve^{-\theta_1 sH_0}M_V\cdots M_Ve^{-\theta_n sH_0}\,d\theta\\
                &\quad + \frac{(-s)^{N+1}}{(N+1)!}\int_{\Delta_{N+1}} e^{-\theta_0 sH_0}M_V\cdots e^{-\theta_{N}sH_0}M_Ve^{-\theta_{N+1}sH}\,d\theta
    \end{align*}
    where $\Delta_n = \{(\theta_0,\ldots,\theta_n) \colon \theta_0+\cdots+\theta_n = 1\}$ is the $n$-simplex, and $d\theta$ is the normalised Lebesgue measure (so $\int_{\Delta_n}\,d\theta = 1$).
    For each fixed $s>0$, since $V$ is bounded it is easy to see that the remainder term tends to zero in the operator norm as $N\to \infty$, so we have the operator norm convergent Dyson series:
    \begin{equation*}
        e^{-sH} = e^{-sH_0} + \sum_{n=1}^\infty \frac{(-s)^n}{n!}\int_{\Delta_n} e^{-\theta_0 sH_0}M_Ve^{-\theta_1 sH_0}M_V\cdots M_Ve^{-\theta_nsH_0}\,d\theta.
    \end{equation*}
    Taking the Fourier transform, we have:
    \begin{align*}
        e^{-s(M_{|x|}^2+V(-i\nabla))} &= e^{-sM_{|x|^2}} \\
                                      &\quad + \sum_{n=1}^\infty \frac{(-s)^n}{n!}\int_{\Delta_n} e^{-\theta_0sM_{|x|}^2}V(-i\nabla)e^{-\theta_1 sM_{|x|}^2}\cdots V(-i\nabla)e^{-\theta_n sM_{|x|}^2}\,d\theta.
    \end{align*}

    For each $n\geq 1$ the integral:
    \begin{equation*}
        I_n := \int_{\Delta_n} e^{-s\theta_0 M_{|x|}^2}V(-i\nabla)e^{-s\theta_1M_{|x|}^2}V(-i\nabla)\cdots e^{-s\theta_{n}M_{|x|}^2}\,d\theta
    \end{equation*}
    converges in the $C^*$-algebra $\Pi$. Indeed, at each $\theta \in \Delta_n$ the integrand belongs to $\Pi$ and is norm continuous as a function of $\theta$. Since the principal symbol function ${\rm sym}$ is norm continuous, we have:
    \begin{align*}
        {\rm sym}(I_n)(x,\omega) &= \int_{\Delta_n} e^{-s(\theta_0+\cdots+\theta_n)|x|^2}V(\xi)^n\,d\theta\\
                            &= e^{-s|x|^2}V(\xi)^n,\quad x \in {\mathbb{R}}^d,\,\omega \in S^{d-1}.
    \end{align*}
    The norm convergence of the Dyson series and the norm continuity of the principal symbol mapping imply that $e^{-s(M_{|x|}^2+V(-i\nabla))} \in \Pi$, and:
    \begin{align*}
        {\rm sym}(e^{-s(M_{|x|}^2+V(-i\nabla))})(x,\omega) &= e^{-s|x|^2}+\sum_{n=1}^\infty \frac{(-s)^n}{n!}{\rm sym}(I_n)(x,\omega)\\
                                                      &= e^{-s|x|^2}\sum_{n=0}^\infty \frac{(-sV(\xi))^n}{n!}\\
                                                      &= e^{-s(|x|^2+V(\xi))}.
    \end{align*}
    
    Let $g$ be a compactly supported smooth function on ${\mathbb{R}}$ with $g(0) = 1$, and let $\varepsilon> 0$.
    Applying \cite[Theorem 1.5]{Dao1} (and the unitary invariance of the Dixmier trace), we have
    \begin{align*}
        {\rm Tr}_\omega(e^{-sH}M_{\anglb{x}}^{-d}g(\varepsilon H_0)) &= {\rm Tr}_\omega(g(\varepsilon M_{|x|}^2)e^{-s(M_{|x|}^2+V(-i\nabla))}(1-\Delta)^{-\frac{d}{2}})\\
                                                                &= \frac{1}{d(2\pi)^ds^{d/2}}\int_{S^{d-1}} e^{-sV(\xi)}\,d\xi\int_{{\mathbb{R}}^d} g(\varepsilon|x|^2)e^{-|x|^2}\,dx.
    \end{align*}    
    
    As $\varepsilon\to 0$ the integral $\int_{{\mathbb{R}}^d} e^{-|x|^2}g(\varepsilon|x|^2)\,dx$ converges to $\pi^{d/2}$ by the dominated convergence theorem.
    Therefore, for all $s > 0$:
    \begin{equation*}
        \lim_{\varepsilon\to 0} {\rm Tr}_\omega(e^{-sH}M^{-d}_{\anglb{x}}g(s\varepsilon H_0)) = \frac{1}{d(4\pi)^{d/2}}\int_{S^{d-1}} s^{-d/2}e^{-sV(\xi)}\,d\xi.
    \end{equation*}
    We need to show that the limit on the left is:
    $${\rm Tr}_\omega(e^{-sH}M^{-d}_{\anglb{x}}).$$
    To achieve this, write $e^{-sH} = e^{-sH/2}e^{-sH/2}$ and use the cyclic property of ${\rm Tr}_\omega$
    \begin{equation*}
        {\rm Tr}_\omega(e^{-sH}M^{-d}_{\anglb{x}}g(s\varepsilon H_0)) = {\rm Tr}_\omega(e^{-sH/2}M^{-d}_{\anglb{x}}g(s\varepsilon H_0)e^{-sH/2}).
    \end{equation*}
    The operator $e^{-sH/2}M^{-d}_{\anglb{x}}$ is in the ideal ${\mathcal{L}}_{1,\infty}$ (this is a consequence of Theorem \ref{main_cwikel_estimate}, to be proved in Section \ref{cwikel_section}). We claim that:
    \begin{equation*}
        \lim_{\varepsilon\to 0} g(s\varepsilon H_0)e^{-sH/2} = e^{-sH/2}
    \end{equation*}
    in the uniform norm. To see this, note that the operator $(H+i)e^{-sH/2}$ is bounded, and the resolvent identity implies that $(H_0+i)(H+i)^{-1}$ is bounded (see the proof of Theorem \ref{main_cwikel_estimate}), so it follows that $(H_0+i)e^{-sH/2}$ is also bounded. By functional calculus for self-adjoint operators, we have:
    \begin{equation*}
        \lim_{\varepsilon\to 0} g(s\varepsilon H_0)(H_0+i)^{-1} = (H_0+i)^{-1}
    \end{equation*}
    in the operator norm. Hence as $\varepsilon\to 0$ we have
    \begin{align*}
        g(s\varepsilon H_0)e^{-sH/2} &= g(s\varepsilon H_0)(H_0+i)^{-1}(H_0+i)e^{-sH/2}\\
                                     &\rightarrow (H_0+i)^{-1}(H_0+i)e^{-sH/2}\\
                                     &= e^{-sH/2}
    \end{align*}
    in the operator norm. Therefore as $\varepsilon\to 0$,
    \begin{equation*}
        e^{-sH/2}M^{-d}_{\anglb{x}}g(s\varepsilon H_0)e^{-sH/2}\rightarrow e^{-sH/2}M^{-d}_{\anglb{x}}e^{-sH/2}
    \end{equation*}
    in the ${\mathcal{L}}_{1,\infty}$ topology. 
    Since ${\rm Tr}_\omega$ is continuous,
    \begin{equation*}
        \lim_{\varepsilon\to 0}{\rm Tr}_\omega(e^{-sH/2}M^{-d}_{\anglb{x}}g(\varepsilon sH_0)e^{-sH/2}) = {\rm Tr}_\omega(e^{-sH/2}M^{-d}_{\anglb{x}}e^{-sH/2}) = {\rm Tr}_\omega(e^{-sH}M^{-d}_{\anglb{x}}).
    \end{equation*}
    
   Therefore we have proved that:
   \begin{equation*}
       {\rm Tr}_\omega(e^{-sH}M^{-d}_{\anglb{x}}) = \frac{1}{d(4\pi s)^{d/2}}\int_{S^{d-1}} e^{-sV(\xi)}\,d\xi.
   \end{equation*}
    
    Hence, \emph{assuming that the density of states $\nu_H$ for $H$ exists}, it follows from Theorem \ref{T: main} that:
    \begin{align*}
        \int_{{\mathbb{R}}} e^{-s\lambda}\,d\nu_{H}(\lambda) & =  \frac d{\omega_d}{\rm Tr}_\omega \brs{ e^{-sH} M^{-d}_{\anglb{x}}}\\
                                                     & = \frac{1}{\omega_d}\int_{S^{d-1}} (4\pi s)^{-d/2}e^{-sV(\xi)}\,d\xi.
    \end{align*}
    Recall from Example \ref{E: 1} that the Laplace transform of the density of states for $H_0$ is given by:
    \begin{equation*}
        \int_{{\mathbb{R}}} e^{-s\lambda}\,d\nu_{H_0}(\lambda) = (4\pi s)^{-d/2}.
    \end{equation*}
    Thus for a constant perturbation $H_0+c$, where $c \in {\mathbb{R}}$, we have
    \begin{equation*}
        \int_{{\mathbb{R}}} e^{-s\lambda}\,d\nu_{H_0+c}(\lambda) = (4\pi s)^{-d/2}e^{-sc}.
    \end{equation*}
    Therefore,
    \begin{equation*}
        \int_{{\mathbb{R}}} e^{-s\lambda}\,d\nu_H(\lambda) = \frac{1}{\omega_d} \int_{S^{d-1}}\int_{{\mathbb{R}}} e^{-s\lambda}\, d\nu_{H_0+V(\xi)}(\lambda)d\xi.
    \end{equation*}
    Finally, Fubini's theorem and the uniqueness of the Laplace transform (c.f. Remark \ref{laplace_uniqueness}) yields
    \begin{equation*}
        \nu_H = \frac{1}{\omega_d}\int_{S^{d-1}} \nu_{H_0+V(\xi)}\,d\xi
    \end{equation*}
    in the sense that if $I\subseteq {\mathbb{R}}$ is Borel, then we have
    \begin{equation*}
        \nu_H(I) = \frac{1}{\omega_d}\int_{S^{d-1}} \nu_{H_0+V(\xi)}(I)\,d\xi.
    \end{equation*}
    
\end{proof}

The result of Theorem \ref{homogeneous_potential_theorem} can be put into a more explicit form by observing that
the density of states for a Hamiltonian $H_0+c$ where $c \in {\mathbb{R}}$ is given by:
\begin{equation*}
    \nu_{H_0+c}((-\infty,t]) = \frac{\omega_d}{d(2\pi)^d}(t-c)_+^{\frac{d}{2}}
\end{equation*}
where $t_+ = \max\{t,0\}$.
Hence if $V$ is positively homogeneous and $\nu_{H_0+V}$ exists, then the result of Theorem \ref{homogeneous_potential_theorem} states that
\begin{equation*}
    \nu_{H_0+V}((-\infty,t]) = \frac{1}{d(2\pi)^d} \int_{S^{d-1}} (t-V(\xi))_+^{\frac{d}{2}}\,d\xi.
\end{equation*}
    
\subsection{Stability of the density of states}
As another example of the utility of Theorem \ref{T: main}, we also show how the theorem can be used to prove the stability of the density of states under ``small" perturbations.
Let $V \in L_\infty({\mathbb{R}}^d)$ be real valued, and let 
\begin{equation*}
    H := -\Delta+M_V
\end{equation*}
be the corresponding Schr\"odinger operator. We are interested in perturbations $H+M_{V_0}$ where $V_0 \in L_\infty({\mathbb{R}}^d)$ is such that
$$M_{V_0}e^{s\Delta}$$
is \emph{compact} on $L_2({\mathbb{R}}^d)$ for any $s > 0$. We have the following:
\begin{thm}\label{stability_of_dos}
    If the density of states measure exists for both $H$ and $H+M_{V_0}$, then the two measures are equal.
\end{thm}
This theorem reflects the stability of the density of states under ``localised" perturbations. Similar statements are already known in the literature \cite[Theorem C.7.7]{SimonSemigroups}, but under different assumptions
on $V$ and $V_0$.

To give some feeling for the class of potentials which satisfy the stated condition, the following is a sufficient condition.
\begin{lemma}\label{tau_compact_lemma}
    Let $V_0 \in L_\infty({\mathbb{R}}^d)$ be a potential such that for all $t > 0$, 
    \begin{equation*}
        |\{x\in {\mathbb{R}}^d \;:\;|V_0(x)|\geq t\}| < \infty.
    \end{equation*}
    Then $M_{V_0}e^{s\Delta}$ is compact for all $s > 0$.
\end{lemma}
\begin{proof}
    Let $\varepsilon> 0$, and let
    \begin{equation*}
        A_{\varepsilon} = \{x \in {\mathbb{R}}^d\;:\;|V_0(x)| \geq \varepsilon\}. 
    \end{equation*}
    Then $\|(1-\chi_{A_{\varepsilon}})V_0\|_{\infty} \leq \varepsilon$, and hence:
    \begin{equation*}
        \|M_{V_0}e^{s\Delta}-M_{\chi_{A_{\varepsilon}}V_0}e^{s\Delta}\|_{\infty} = \|M_{(1-\chi_{A_{\varepsilon}})V}e^{s\Delta}\|_\infty \leq \varepsilon\|e^{s\Delta}\|_{\infty} = \varepsilon.
    \end{equation*}
    Therefore in the operator norm we have
    $$\lim_{\varepsilon\to 0} M_{\chi_{A_{\varepsilon}}}M_{V_0}e^{s\Delta} = M_{V_0}e^{s\Delta}.$$
    
    On the other hand, $A_{\varepsilon}$ has finite measure and $V_0$ is bounded, so it follows that the function $\chi_{A_{\varepsilon}}V_0$ is square integrable. Hence
    the operator $M_{\chi_{A_{\varepsilon}}V_0}e^{s\Delta}$ is Hilbert-Schmidt and in particular, compact (see our review of classical Cwikel-type estimates below in Section \ref{cwikel_section}). Thus, $M_{V_0}e^{s\Delta}$ is the limit in the operator norm of a sequence of compact operators, so is itself compact.
\end{proof}

Recall that $H = -\Delta+M_V$ is an arbitrary Schr\"odinger operator with bounded real-valued potential $V$.
\begin{lemma}\label{first_compact_perturbation_lemma}
    If $V_0 \in L_\infty({\mathbb{R}}^d)$ is such that $M_{V_0}e^{s\Delta}$ is compact for all $s > 0$, then $M_{V_0}e^{-sH}$ is also compact for all $s > 0$.
\end{lemma}
\begin{proof}
    By taking the adjoint, it suffices to check that $e^{-sH}M_{V_0}$ is compact. Using Duhamel's formula:
    \begin{equation*}
        e^{-sH} = e^{s\Delta}-s\int_0^1 e^{-s\theta H}M_{V}e^{s(1-\theta)\Delta}\,d\theta
    \end{equation*}
    It follows that
    \begin{equation*}
        e^{-sH}M_{V_0}-e^{s\Delta}M_{V_0} = -s\int_0^1 e^{-s\theta H}M_Ve^{s(1-\theta)\Delta}M_{V_0}\,d\theta.
    \end{equation*}
    We will complete the proof by showing that the integral on the right hand side
    is a convergent ${\mathcal{K}}(L_2({\mathbb{R}}^d))$-valued Bochner integral, and hence in particular is an element of ${\mathcal{K}}(L_2({\mathbb{R}}^d))$. Note that the integrand belongs to ${\mathcal{K}}(L_2({\mathbb{R}}^d))$ for each $\theta \in (0,1)$,
    and has uniformly bounded operator norm.
    Since ${\mathcal{K}}(L_2({\mathbb{R}}^d))$ is separable, the Bochner theorem implies that to ensure Bochner integrability it is enough to show that $\theta \mapsto e^{-s\theta H}M_{V}e^{-s(1-\theta)\Delta}M_{V_0}$ is weakly measurable. Since the semigroups $e^{-s\theta H}$ and $e^{s(1-\theta)\Delta}$ are strongly continuous and uniformly bounded, it follows that the integrand is weakly continuous and hence weakly measurable. Thus the integral defines an element of ${\mathcal{K}}(L_2({\mathbb{R}}^d))$, and
    this completes the proof.
\end{proof}

\begin{lemma}\label{second_compact_perturbation_lemma}
    If $V_0 \in L_\infty({\mathbb{R}}^d)$ is such that $M_{V_0}e^{s\Delta}$ is compact for all $s > 0$, then the difference
    $$e^{-s(H+M_{V_0})}-e^{-sH}$$
    is compact for all $s > 0$.
\end{lemma}
\begin{proof}
    Using Duhamel's formula:
    \begin{equation*}
        e^{-s(H+M_{V_0})} = e^{-sH}-s\int_0^1 e^{-s\theta (H+M_{V_0})}M_{V_0}e^{-s(1-\theta)H}\,d\theta.
    \end{equation*}
    Lemma \ref{first_compact_perturbation_lemma} implies that the integrand $e^{-s\theta (H+M_{V_0})}M_{V_0}e^{-s(1-\theta)H}$ is compact for each $\theta \in [0,1]$
    and all $s > 0$. Similar reasoning to the proof of Lemma \ref{first_compact_perturbation_lemma} implies that the integral converges as a ${\mathcal{K}}(L_2({\mathbb{R}}^d))$-valued Bochner integral,
    and hence the difference $e^{-s(H+M_{V_0})}-e^{-sH}$ is compact.
\end{proof}

We now prove the equality of the left hand sides of \eqref{F: phi(..)=mu(g)} for the potentials $H$ and $H+M_{V_0}$.
\begin{prop}
    Suppose that $H = -\Delta+M_V$ is a Schr\"odinger operator with bounded real-valued potential $V$, and let $V_0 \in L_\infty({\mathbb{R}}^d)$ be a real-valued potential such that $M_{V_0}e^{s\Delta}$
    is compact for all $s > 0$.
    Then for all Dixmier traces ${\rm Tr}_\omega$ and all $s > 0$ we have:
    \begin{equation*}
        {\rm Tr}_\omega(e^{-sH}(1+M_{|x|}^2)^{-d/2}) = {\rm Tr}_\omega(e^{-s(H+M_{V_0})}(1+M_{|x|}^2)^{-d/2}).
    \end{equation*}
\end{prop}
\begin{proof}
    It follows from Theorem \ref{main_cwikel_estimate}, to be proved below, that the operators $e^{-sH}(1+M_{|x|}^2)^{-d/2}$ and $e^{-s(H+M_{V_0})/2}(1+M_{|x|}^2)^{-d/2}$ individually belong to ${\mathcal{L}}_{1,\infty}$
    so each side of the equality is meaningful. Using the cyclic property of the Dixmier trace, we have:
    \begin{equation*}
        {\rm Tr}_\omega(e^{-sH}(1+M_{|x|}^2)^{-d/2}) = {\rm Tr}_\omega(e^{-sH/2}(1+M_{|x|}^2)^{-d/2}e^{-sH/2})
    \end{equation*}
    and similarly
    \begin{equation*}
        {\rm Tr}_\omega(e^{-s(H+M_{V_0}}(1+M_{|x|}^2)^{-d/2}) = {\rm Tr}_\omega(e^{-s(H+M_{V_0})/2}(1+M_{|x|}^2)^{-d/2}e^{-s(H+M_{V_0})/2}).
    \end{equation*}
    For each $s$, we have the identity:
    \begin{align*}
        e^{-s(H+M_{V_0})/2}(1+M_{|x|}^2)^{-d/2}&e^{-s(H+M_{V_0})/2}-e^{-sH/2}(1+M_{|x|}^2)^{-d/2}e^{-sH/2}\\
                                               &= (e^{-s(H+M_{V_0})/2}-e^{-sH/2})(1+M_{|x|}^2)^{-d/2}e^{-s(H+M_{V_0})/2}\\
                                               &\quad +e^{-sH/2}(1+M_{|x|}^2)^{-d/2}(e^{-sH_0/2}-e^{-sH/2}).
    \end{align*}
    The operators $e^{-sH/2}(1+M_{|x|}^2)^{-d/2}$ and $(1+M_{|x|}^2)^{-d/2}e^{-s(H+M_{V_0})}$ belong to ${\mathcal{L}}_{1,\infty}$, and Lemma \ref{second_compact_perturbation_lemma} implies
    that the difference $e^{-s(H+M_{V_0})/2}-e^{-sH/2}$ is compact. Hence,
    $$
        e^{-s(H+M_{V_0})/2}(1+M_{|x|}^2)^{-d/2}e^{-s(H+M_{V_0})/2}-e^{-sH/2}(1+M_{|x|}^2)^{-d/2}e^{-sH/2} \in {\mathcal{K}}\cdot {\mathcal{L}}_{1,\infty}+{\mathcal{L}}_{1,\infty}\cdot {\mathcal{K}}.
    $$
    The product of a compact operator and a weak trace-class operator belongs to $({\mathcal{L}}_{1,\infty})_0$ -- the separable part of the ideal ${\mathcal{L}}_{1,\infty}$ -- 
    and in particular is in the kernel of every Dixmier trace. Thus,
    \begin{equation*}
        {\rm Tr}_\omega(e^{-s(H+M_{V_0})/2}(1+M_{|x|}^2)^{-d/2}e^{-s(H+M_{V_0})/2}) = {\rm Tr}_\omega(e^{-sH/2}(1+M_{|x|}^2)^{-d/2}e^{-sH/2}).
    \end{equation*}
    Once again using the cyclic property of the trace, the result immediately follows.    
\end{proof}

The proof of Theorem \ref{stability_of_dos} now follows immediately from Theorem \ref{T: main} and the uniqueness property of the Laplace transform (see Remark \ref{laplace_uniqueness}), if we assume the existence of the density of states for both $H$ and $H+M_{V_0}$.

\subsection{Asymptotically homogeneous potentials}

As a straightforward combination of the preceding two sections, we can also present a formula for potentials $V\in L_\infty({\mathbb{R}}^d)$ such that there exists a ``uniform radial limit at infinity" in the sense that for all $x \in {\mathbb{R}}^d\setminus \{0\}$, the limit:
    \begin{equation*}
        V_h(x) := \lim_{r\to\infty} V(rx)
    \end{equation*}
    exists, and converges uniformly over $x \in S^{d-1}$. The function $V_h$ is positively homogeneous in the sense
    that $V_h(tx) = V_h(x)$ for all $x\in {\mathbb{R}}^d\setminus \{0\}$ and $t > 0$.
    Then if the DOS measures $\nu_{-\Delta+V}$ and $\nu_{-\Delta+V_h}$ exist, $\nu_{-\Delta+V}$ is given by the formula:
    \begin{equation*}
        \nu_{-\Delta+V} = \frac{1}{\omega_d} \int_{S^{d-1}} \nu_{-\Delta+V_h(\xi)}\,d\xi.
    \end{equation*}
Indeed, the assumption of uniform convergence on $x \in S^{d-1}$ of the limit $V_h(x) = \lim_{r\to\infty} V(rx)$ implies that $V-V_h$ satisfies the assumption of Lemma \ref{tau_compact_lemma}, and hence by Theorem \ref{stability_of_dos} it follows that $V$ and $V_h$ have the same DOS measure. Since $V_h$ is positively homogeneous, the formula for $\nu_{-\Delta+V}$ then follows from Theorem \ref{homogeneous_potential_theorem}.

\section{Preliminaries}
\subsection{Trace ideals}
The following material is standard; for more details we refer the reader to \cite{LSZ2012,Simon1979}.
Let ${\mathcal{H}}$ be a complex separable Hilbert space, and let ${\mathcal{B}}({\mathcal{H}})$ denote the set of all bounded operators on ${\mathcal{H}}$, and let ${\mathcal{K}}({\mathcal{H}})$ denote the ideal of compact operators on ${\mathcal{H}}$. Given $T\in {\mathcal{K}}({\mathcal{H}})$, the sequence of singular values 
$\mu(T) = \{\mu(k,T)\}_{k=0}^\infty$ is defined as:
\begin{equation*}
    \mu(k,T) = \inf\{\|T-R\|\;:\;\mathrm{rank}(R) \leq k\}.
\end{equation*}
Let $p \in (0,\infty).$ The Schatten class ${\mathcal{L}}_p$ is the set of operators $T$ in ${\mathcal{K}}({\mathcal{H}})$ such that $\mu(T)$ is $p$-summable, i.e. in the sequence space $\ell_p$. If $p \geq 1$ then the ${\mathcal{L}}_p$
norm is defined as:
\begin{equation*}
    \|T\|_p := \|\mu(T)\|_{\ell_p} = \left(\sum_{k=0}^\infty \mu(k,T)^p\right)^{1/p}.
\end{equation*}
With this norm ${\mathcal{L}}_p$ is a Banach space, and an ideal of ${\mathcal{B}}({\mathcal{H}})$.

The weak Schatten class ${\mathcal{L}}_{p,\infty}$ is the set of operators $T$ such that $\mu(T)$ is in the weak $L_p$-space $\ell_{p,\infty}$, with quasi-norm:
\begin{equation*}
    \|T\|_{p,\infty} = \sup_{k\geq 0} (k+1)^{1/p}\mu(k,T) < \infty.
\end{equation*}
As with the ${\mathcal{L}}_p$ spaces, ${\mathcal{L}}_{p,\infty}$ is an ideal of ${\mathcal{B}}({\mathcal{H}})$. We also have the following form
of H\"older's inequality,
\begin{equation}\label{weak holder}
    \|TS\|_{r,\infty} \leq c_{p,q}\|T\|_{p,\infty}\|S\|_{q,\infty}
\end{equation}
where $\frac{1}{r}=\frac{1}{p}+\frac{1}{q}$, for some constant $c_{p,q}$. Indeed, this follows from the definition of the weak ${\mathcal{L}}_p$-quasinorms
and the inequality $\mu(2n,TS)\leq \mu(n,T)\mu(n,S)$ for $n\geq 1$ \cite[Proposition 1.6]{Fack1982}, \cite[Corollary 2.2]{GohbergKrein}.

Note that if $r > p$, then we have the inequality:
\begin{equation}\label{zeta_function_inequality}
    \|T\|_{r}^r = \sum_{k=0}^\infty \mu(k,T)^r \leq \sum_{k=0}^\infty (k+1)^{-\frac{r}{p}}\|T\|_{p,\infty}^r = \zeta\left(\frac{r}{p}\right)\|T\|_{p,\infty}^r
\end{equation}
where $\zeta$ is Riemann's zeta function.

For $q \in [1,\infty)$, we also consider the ideal ${\mathcal{L}}_{q,1}$, defined as the set of compact operators $T$ on ${\mathcal{H}}$ satisfying:
\begin{equation*}
    \|T\|_{{\mathcal{L}}_{q,1}} := \sum_{n\geq 0} \frac{\mu(n,T)}{(n+1)^{1-\frac{1}{q}}} < \infty.
\end{equation*}
We have the following H\"older-type inequality: if $\frac{1}{p}+\frac{1}{q} = 1$ then
\begin{equation}\label{weird_holder}
    \|TS\|_1 \leq \|T\|_{p,\infty}\|S\|_{q,1}
\end{equation}
(see e.g. \cite[p. 303]{Connes}).

For this paper, the relevant continuous embeddings between these ideals are
\begin{equation}\label{L_p_embeddings}
    {\mathcal{L}}_{p,\infty}\subset {\mathcal{L}}_{q},\quad{\mathcal{L}}_{p,\infty}\subset {\mathcal{L}}_{q,1},\quad 0 < p < q \leq \infty
\end{equation}
(see e.g. \cite[\S IV.2.$\alpha$, Proposition 1]{Connes}).

Among ideals of particular interest is ${\mathcal{L}}_{1,\infty}$, and we are concerned with traces on this ideal. For more details, see \cite[Section 5.7]{LSZ2012} and \cite{SSUZ2015}. A linear functional $\varphi:{\mathcal{L}}_{1,\infty}\to {\mathbb{C}}$ is called a trace if it is unitarily invariant. That is, for all unitary operators $U$ and for all $T\in {\mathcal{L}}_{1,\infty}$ we have that $\varphi(U^*TU) = \varphi(T)$. It follows that for all bounded operators $B$ we have $\varphi(BT)=\varphi(TB).$   

A Dixmier trace ${\rm Tr}_\omega$ is a trace on ${\mathcal{L}}_{1,\infty}$ defined in terms of an extended limit $\omega \in \ell_\infty(\mathbb{N})^*$ (i.e., a continuous extension of the limit functional to $\ell_\infty(\mathbb{N})$).
Given a positive operator $T \in {\mathcal{L}}_{1,\infty}$, ${\rm Tr}_\omega(T)$ is defined as:
\begin{equation*}
    {\rm Tr}_{\omega}(T) = \omega\left(\left\{\frac{\sum_{k=0}^N\mu(k,T)}{\log(2+N)}\right\}_{N=0}^\infty\right).
\end{equation*}
If $\omega$ is \emph{dilation invariant}, that is, if for all $n\geq 1$ we have $\omega\circ \sigma_n = \omega$, where $\sigma_n$ is the dilation
semigroup $\sigma_n(\{a_j\}_{j=0}^\infty) = \{a_{\lfloor \frac{j}{n}\rfloor}\}_{j=0}^\infty$, then ${\rm Tr}_\omega$ is called a Dixmier trace and extends to a linear functional
on ${\mathcal{L}}_{1,\infty}$.

We note that it can however be proved that ${\rm Tr}_\omega$ extends to a trace on ${\mathcal{L}}_{1,\infty}$ with no extra invariance conditions on $\omega$ (see \cite[Theorem 17]{SS2013})\footnote{Moreover it can be proved that ${\rm Tr}_\omega$ is a Dixmier
trace for every extended limit $\omega$. This will appear in the upcoming second edition of \cite{LSZ2012}.}.

More generally, an extended limit is a bounded linear functional $\omega$ on $L_\infty((0,\infty))$ which extends the limit functional from the subspace of functions having limit at~$\infty$ to all of $L_\infty((0,\infty))$.

\subsection{Double operator integrals}
In this paper we will make brief use of the technique of double operator integrals for unitary operators. See e.g. \cite{AleksandrovPeller,ABF,ACS,BirmanSolomyak1989,BirmanSolomyakSurvey,dPSW}.

Given two unitary operators $U$ and $V$ on ${\mathcal{H}}$, a double operator integral
with symbol $\phi \in L_\infty(\mathbb{T}^2)$ is a linear map $T^{U,V}_{\phi}:{\mathcal{L}}_{2}\to {\mathcal{L}}_{2}$ defined as follows. The operators
$U$ and $V$ also act as unitary operators of left and right multiplication on the Hilbert-Schmidt space ${\mathcal{L}}_{2}$:
\begin{equation*}
    L_UX = UX,\quad R_VX = XV, \quad X \in {\mathcal{L}}_{2}.
\end{equation*}
As linear operators on ${\mathcal{L}}_2$, $L_U$ and $R_V$ are commuting unitary operators and hence there is a joint functional calculus
$\phi\mapsto \phi(L_U,L_V) \in {\mathcal{B}}({\mathcal{L}}_2)$ for $\phi$ a bounded function on the torus $\mathbb{T}^2$. Denote $T^{U,V}_{\phi} := \phi(L_U,R_V)$.
For a Lipschitz class function $f$ on~$\mathbb{T}$, denote by $f^{[1]}$ the divided difference function $f^{[1]}(z,w) = \frac{f(z)-f(w)}{z-w}$ set to an arbitrary value on the diagonal. 

A short computation based on a Fourier decomposition of $f$ (see \cite[Theorem 1.1.3]{AleksandrovPeller}) shows that 
\begin{equation}\label{naive_bound}
    \|T^{U,V}_{f^{[1]}}|_{{\mathcal{L}}_1}\|_{{\mathcal{L}}_1\to {\mathcal{L}}_1} \leq \|\widehat{f'}\|_{\ell_1(\mathbb{Z})} \leq \|f'\|_{L_2(\mathbb{T})}+\|f''\|_{L_2(\mathbb{T})}.
\end{equation}
Provided the above right hand side is finite, we also have that $T^{U,V}_{f^{[1]}}$ extends by duality to ${\mathcal{B}}({\mathcal{H}})$, and an interpolation argument (as described in e.g. \cite[p. 5225]{APPS}) implies that
if $p > 1$ we have
$$\|T^{U,V}_{f^{[1]}}\|_{{\mathcal{L}}_{p,\infty}\to {\mathcal{L}}_{p,\infty}} \leq \|f'\|_{L_2(\mathbb{T})}+\|f''\|_{L_2(\mathbb{T})}$$
and similarly if $q > 1$ we have
\begin{equation}\label{interpolated_naive_bound}
    \|T^{U,V}_{f^{[1]}}\|_{{\mathcal{L}}_{q,1}\to {\mathcal{L}}_{q,1}} \leq \|f'\|_{L_2(\mathbb{T})}+\|f''\|_{L_2(\mathbb{T})}.
\end{equation}

\noindent If $X \in {\mathcal{B}}({\mathcal{H}})$, then we also have the following identity (see \cite[Theorem 8.5]{BirmanSolomyakSurvey}, \cite[Theorem 3.5.4]{AleksandrovPeller} or \cite[Theorem 4.1]{BirmanSolomyak1989} for the related self-adjoint case):
\begin{equation}\label{DK_formula}
    T^{U,U}_{f^{[1]}}([U,X]) = [f(U),X].
\end{equation}

If $g$ is a bounded function on the spectrum of $U$, then it follows from the definition of $T^{U,U}_{f^{[1]}}$ that we also have:

\begin{equation}\label{compatibility}
    g(U)T^{U,U}_{f^{[1]}}(X) = T^{U,U}_{f^{[1]}}(g(U)X),\quad T^{U,U}_{f^{[1]}}(X)g(U) = T^{U,U}_{f^{[1]}}(Xg(U)),\quad X \in {\mathcal{B}}({\mathcal{H}}).
\end{equation}

\section{Cwikel type estimates}\label{cwikel_section}
We will extensively use the notation 
$$
\anglb{x} = (1+\abs{x}^2)^{1/2}
$$
for $x \in \mathbb{R}^d$ and $\abs{x}$ denotes the $\ell_2$-norm of $x$, so that $x\mapsto \anglb{x}^{-1} \in L_{d,\infty}(\mathbb{R}^d)$. Recall that $V$ is a bounded measurable real valued function on $\mathbb{R}^d$, and $H = -\Delta+M_V$ is the Schr\"odinger operator associated to the potential $V$.
We exclusively consider $d\geq 2$.

This section is devoted to a proof of the claim that for integers $p\geq 1$ and $z \in {\mathbb{C}}\setminus {\mathbb{R}}$, we have that $(H+z)^{-p}M^{-p}_{\anglb{x}}$ is in the ideal ${\mathcal{L}}_{d/p,\infty}$.
This is a crucial component to proving that the operator inside the Dixmier trace in Theorem~\ref{T: main} is indeed in the ideal ${\mathcal{L}}_{1,\infty}$. Somewhat similar estimates on the singular values of operators of the form
$f(H)M_g$ are also in \cite[Section B.9]{SimonSemigroups}.

Our proofs are based on the following classical Cwikel estimate (see \cite[Theorem 4.2]{Simon1979} or for the $p=2$ case see the more recent \cite[Corollary 1.2, Theorem 5.6]{LeSZ_cwikel}).

The function spaces $\ell_{2,\infty}(L_4)(\mathbb{R}^d)$ and $\ell_{2,\log}(L_\infty)(\mathbb{R}^d)$ are defined by the norms:  
\begin{align*}
        \|g\|_{\ell_{2,\infty}(L_4)(\mathbb{R}^d)} &= \|\{\|g\|_{L_4(k+[0,1]^d)}\}_{k\in \mathbb{Z}^d}\|_{\ell_{2,\infty}},\\
    \|f\|_{\ell_{2,\log}(L_\infty)(\mathbb{R}^d)} &= \left(\sum_{k\in \mathbb{Z}^d} (1+\log(1+|k|))\|f\|_{L_\infty(k+[0,1]^d)}^2\right)^{1/2} 
\end{align*}

where
$|k|$ denotes the $\ell_2$-norm of $k \in \mathbb{Z}^d$.
\begin{prop}\label{classical_cwikel}
    For $2 < p < \infty,$ if $f \in L_p({\mathbb{R}}^d)$ and $g \in L_{p,\infty}({\mathbb{R}}^d)$, then the operator $M_f g(-i\nabla)$ is in the ideal ${\mathcal{L}}_{p,\infty}$.
   If $g \in \ell_{2,\infty}(L_4)(\mathbb{R}^d)$ and $f \in \ell_{2,\log}(L_\infty)(\mathbb{R}^d),$
    then $M_fg(-i\nabla) \in {\mathcal{L}}_{2,\infty}(\mathbb{R}^d).$ 
\end{prop}

We begin with a lemma of elementary operator theory, required for the proof of the main result of this section (Theorem~\ref{main_cwikel_estimate}).
\begin{lemma}\label{inclusion lemma} 
Let $A,B,C$ be bounded operators such that $A=B-AC.$ If $B\in\mathcal{L}_{p_0,\infty}$ and $C\in\mathcal{L}_{p_1,\infty},$ for $0 < p_0,p_1 < \infty$ then $A\in\mathcal{L}_{p_0,\infty}.$
\end{lemma}
\begin{proof}
By induction, for each $n \geq 1$ we have:
$$A=\sum_{k=0}^{n-1}(-1)^kBC^k+(-1)^nAC^n.$$
Since ${\mathcal{L}}_{p_0,\infty}$ is an ideal, for all $k \geq 0$ we have 
$BC^k\in\mathcal{L}_{p_0,\infty}.$
Choose $n$ sufficiently large such that $p_1<np_0.$ From H\"older's inequality \eqref{weak holder}, it follows that
$$AC^n\in\mathcal{L}_{\frac{p_1}{n},\infty}\subset \mathcal{L}_{p_0,\infty}.$$
Hence, $A\in\mathcal{L}_{p_0,\infty}.$
\end{proof}

The following lemma contains a crucial piece of the proof of Theorem~\ref{main_cwikel_estimate}, the main result in this section. For uniformity of notations, we set $H_0=-\Delta.$

\begin{lemma}\label{first commutator lemma} 
For all integers $p \geq 0$ and for any ${\varepsilon}>0,$ we have:
$$M^{p}_{\anglb{x}}   \SqBrs{  H_0, M^{-p}_{\anglb{x}}  } (H_0+1)^{-1}\in\mathcal{L}_{d+{\varepsilon},\infty}.$$
\end{lemma}
\begin{proof} 
 By definition, on the space ${\mathcal{S}}({\mathbb{R}}^d)$ of Schwartz-class functions, we have
$$H_0=-\sum_{m=1}^d\partial_m^2.$$
If $f \in C^\infty({\mathbb{R}}^d)$ is bounded, with all derivatives up to all order bounded, then a straightforward calculation shows that 
\begin{equation*} 
    \begin{split}
    [\partial_m^2,M_f] 
                        =  2M_{\partial_mf}\partial_m+M_{\partial_m^2f},\quad 1\leq m\leq d
    \end{split}
\end{equation*}
which gives 
\begin{equation}\label{leibniz}
    [H_0, M_f] = -2 \sum_{m=1}^d  M_{\partial_m f} \partial_m - M_{H_0 f}.
\end{equation}
 which is also valid on ${\mathcal{S}}({\mathbb{R}}^d)$.

\noindent
Define the function $f_p$ by
$f_p(x) : =\anglb{x}^{-p}, \ x\in\mathbb{R}^d.$
    For $1\leq m\leq d$, we have (here $x=(x_1,\cdots,x_d)\in\mathbb{R}^d$):
\begin{align}
    \partial_m f_p(x) &= -px_m\anglb{x}^{-p-2} \label{derivative_computations},\\
    \partial_m^2f_p(x) &= p(p+2)x_m^2\anglb{x}^{-p-4}- p\anglb{x}^{-p-2}\nonumber,\\
            H_0f_p(x) &= - p(p+2-d)\anglb{x}^{-p-2}\nonumber. 
\end{align}
Using \eqref{leibniz} (once again on the domain ${\mathcal{S}}({\mathbb{R}}^d)$), we write
\begin{align} 
        M^p_{\anglb{x}}  [H_0,M^{-p}_{\anglb{x}}]  &= M_{f_p^{-1}}[H_0,M_{f_p}]\nonumber\\ 
                                                &= M_{f_p^{-1}}(-2\sum_{m=1}^d M_{\partial_mf_p}\partial_m - M_{H_0 f_p})\nonumber\\
                                                &= - 2\sum_{m=1}^dM_{f_p^{-1}\partial_mf_p}\partial_m - M_{f_p^{-1}H_0f_p}.\label{F: aaa(1)}
\end{align}

Since for any ${\varepsilon}>0$ 
$$ 
    f_p^{-1}\partial_mf_p(x) = - p x_m \anglb{x}^{-2} \in L_{d,\infty}({\mathbb{R}}^d) \cap L_\infty({\mathbb{R}}^d) \subset L_{d+{\varepsilon}}({\mathbb{R}}^d)
$$ 
and  $\partial_m (H_0+1)^{-1} = g(-i\nabla)$ with 
$$
g(t) = \frac{-it_m}{1+|t|^2} \in L_{d,\infty}({\mathbb{R}}^d) \cap L_\infty({\mathbb{R}}^d) \subset L_{d+{\varepsilon},\infty}({\mathbb{R}}^d) 
$$ 
it follows from Proposition~\ref{classical_cwikel}
that for any ${\varepsilon}>0$
\begin{align} \label{F: aaa(2)}
    \sum_{m=1}^pM_{f_p^{-1}\partial_mf_p}\partial_m (H_0+1)^{-1} \in {\mathcal{L}}_{d+{\varepsilon},\infty}.
\end{align}
Since 
$$ 
    f_p^{-1} H_0 f_p = - p(p+2-d)\anglb{x}^{-2} \in L_{d/2,\infty}({\mathbb{R}}^d) \cap L_\infty({\mathbb{R}}^d)
$$
and  $(H_0+1)^{-1} = g(-i\nabla)$ with $g(t) = \frac{1}{1+|t|^2} \in L_{d/2,\infty}({\mathbb{R}}^d)\cap L_\infty({\mathbb{R}}^d),$ it follows from Proposition~\ref{classical_cwikel}
that for any ${\varepsilon}>0$
\begin{align} \label{F: aaa(3)}
    M_{f_p^{-1}H_0f_p}  (H_0+1)^{-1} \in  {\mathcal{L}}_{d/2+{\varepsilon},\infty}.
\end{align}
Combining (\ref{F: aaa(1)}), (\ref{F: aaa(2)}) and (\ref{F: aaa(3)}) gives for any ${\varepsilon}>0$
$$
M^p_{\anglb{x}}  [H_0,M^{-p}_{\anglb{x}}]  (H_0+1)^{-1}  \in  {\mathcal{L}}_{d+{\varepsilon},\infty}.
$$
\end{proof}

We now present the main Cwikel estimate of this section:
\begin{thm}\label{main_cwikel_estimate}
For any $z\in\mathbb{C}\setminus\mathbb{R}$ and $p=1,2,\ldots$
$$(H+z)^{-p}M^{-p}_{\anglb{x}}\in \mathcal{L}_{d/p,\infty}.$$
\end{thm}
\begin{proof} 
We introduce the notation:
\begin{align*}
    A_p(z) &= (H+z)^{-p}M^{-p}_{\anglb{x}},\\
    B_p(z) &= (H+z)^{1-p}M^{-p}_{\anglb{x}}(H+z)^{-1},\\
    C_p(z) &= M^{p}_{\anglb{x}}[H_0,M^{-p}_{\anglb{x}}](H+z)^{-1}.
\end{align*}
We prove the assertion by induction on $p$, our goal being to prove that $A_p(z) \in {\mathcal{L}}_{d/p,\infty}$ for all $p\geq 1$.
The resolvent identity gives
$$(H+z)^{-1}(H_0+z)=1-(H+z)^{-1}M_V\in{\mathcal{B}}({\mathcal{H}}).$$
Consider the base case $p=1$. If $d > 2$, then since $x \mapsto \anglb{x}^{-1} \in L_{d,\infty}({\mathbb{R}}^d)$ and $y \mapsto (y^2+z)^{-1} \in L_{d/2,\infty}({\mathbb{R}}^d) \cap L_\infty({\mathbb{R}}^d) \subset L_d({\mathbb{R}}^d),$
the Fourier dual of Proposition~\ref{classical_cwikel} yields:
$$
    (H_0+z)^{-1}M^{-1}_{\anglb{x}} \in {\mathcal{L}}_{d,\infty}.
$$
On the other hand, if $d=2$, then we shall verify that $x \mapsto \anglb{x}^{-1} \in \ell_{2,\infty}(L_4)(\mathbb{R}^2)$,
and that $y \mapsto (|y|^2+z)^{-1} \in \ell_{2,\log}(L_\infty)(\mathbb{R}^2)$. 
There is a constant $C_d$ such that for $k \in \mathbb{Z}^2$ we have:
\begin{equation*}   
    \|\anglb{x}^{-1}\|_{L_4(k+[0,1]^d)} \leq C_d\anglb{k}^{-1}
\end{equation*}
and therefore,
\begin{equation*}
    \|\anglb{x}^{-1}\|_{\ell_{2,\infty}(L_4)(\mathbb{R}^2)} \leq C_d\|\{\anglb{k}^{-1}\}_{k\in \mathbb{Z}^2}\|_{\ell_{2,\infty}} < \infty.
\end{equation*}
Similarly, for $k \in \mathbb{Z}^2$, there is a constant $C_{d,z}$ such that:
\begin{equation*}
    \|(|y|^2+z)^{-1}\|_{L_\infty(k+[0,1]^d)} \leq C_{d,z}\anglb{k}^{-4}
\end{equation*}
and therefore,
\begin{equation*}
    \|(|y|^2+z)^{-1}\|_{\ell_{2,\log}(L_\infty)(\mathbb{R}^2)} \leq C_{d,z}\left(\sum_{k \in \mathbb{Z}^2} (1+\log(1+|k|))\anglb{k}^{-4}\right)^{1/2} < \infty.
\end{equation*}

It then follows from Proposition~\ref{classical_cwikel} that $(H_0+z)^{-1}M^{-1}_{\anglb{x}} \in {\mathcal{L}}_{2,\infty}(L_2(\mathbb{R}^2))$
when $d = 2$. So in all cases $d\geq 2$, we have $(H_0+z)^{-1}M^{-1}_{\anglb{x}} \in {\mathcal{L}}_{d,\infty}$.

Hence,
$$A_1(z)=(H+z)^{-1}(H_0+z)\cdot (H_0+z)^{-1}M^{-1}_{\anglb{x}}  \in  {\mathcal{B}}({\mathcal{H}})\cdot\mathcal{L}_{d,\infty}=\mathcal{L}_{d,\infty}.$$
This proves the $p=1$ case.

Suppose the estimate holds for $p\geq 1.$ Let us prove it for $p+1.$ If $X$ and $Y$ are linear operators where $Y$ and $X^{-1}$ are bounded and $Y$ preserves the domain of $X$, then we have the identity:
\begin{equation}\label{favourite_identity}
[X^{-1},Y]=-X^{-1}[X,Y]X^{-1}.
\end{equation}
This identity applies with $Y = M^{-p-1}_{\anglb{x}}$ and $X = H+z$, since the domain of $X$ is equal to the domain of $H_0$ \cite[Theorem 8.8]{Schmudgen},
and clearly $Y$ preserves the domain of $H_0$.
Thus we may write
\begin{align*}
    A_{p+1}(z)-B_{p+1}(z)&=(H+z)^{-p}\cdot [(H+z)^{-1},M^{-p-1}_{\anglb{x}}]\\
                            &=-(H+z)^{-p-1}[H_0,M^{-p-1}_{\anglb{x}}](H+z)^{-1}\\
                            &=-A_{p+1}(z)C_{p+1}(z).
\end{align*}
We now verify that the operators $A_{p+1}(z),$ $B_{p+1}(z)$ and $C_{p+1}(z)$ satisfy the assumptions in Lemma~\ref{inclusion lemma}.

With the identity
$$B_{p+1}(z)=A_p(z)A_1(\bar{z})^{\ast},$$
H\"older's inequality \eqref{weak holder} and the inductive assumption yield
$$B_{p+1}(z)\in\mathcal{L}_{\frac{d}{p},\infty}\cdot\mathcal{L}_{d,\infty}\subset\mathcal{L}_{\frac{d}{p+1},\infty}.$$
Also,
$$C_{p+1}(z)=M^{p+1}_{\anglb{x}}[H_0,M^{-p-1}_{\anglb{x}}](H_0+z)^{-1}\cdot (H_0+z)(H+z)^{-1}.$$
Lemma~\ref{first commutator lemma} states that:
$$M^{p+1}_{\anglb{x}}[H_0,M^{-p-1}_{\anglb{x}}](H_0+z)^{-1}\in\mathcal{L}_{2d,\infty}.$$
Since
$$(H_0+z)(H+z)^{-1}=1-M_V(H+z)^{-1}\in{\mathcal{B}}({\mathcal{H}})$$
it follows that
$$C_{p+1}(z)\in\mathcal{L}_{2d,\infty}\cdot{\mathcal{B}}({\mathcal{H}})=\mathcal{L}_{2d,\infty}.$$
That is, $B_{p+1}(z) \in {\mathcal{L}}_{d/(p+1),\infty}$ and $C_{p+1}(z) \in {\mathcal{L}}_{2d,\infty}$, so applying Lemma~\ref{inclusion lemma} to the operators $A_{p+1}(z),$ $B_{p+1}(z)$ and $C_{p+1}(z)$
yields $A_{p+1}(z) \in {\mathcal{L}}_{d/(p+1),\infty}$, so the assertion follows by induction on $p$.
\end{proof}

As a useful corollary, we also include the following:
\begin{prop}\label{doi_corollary}
Let ${\varepsilon} > 0$. Then
\begin{equation*}
    \limsup_{r\downarrow d}\|[M^{-r}_{\anglb{x}},e^{-{\varepsilon} (r-1) H}]\|_{1} < \infty
\end{equation*}
and
\begin{equation*}
    \limsup_{r\downarrow d}\|[M^{1-r}_{\anglb{x}},e^{-{\varepsilon} (r-1)H}]\|_{\frac{d}{d-1},1} < \infty.
\end{equation*}
\end{prop}
\begin{proof}   
Let $U$ be the unitary operator $\frac{H+i}{H-i}$, and let $\phi_{{\varepsilon}}$ be a smooth function on the unit circle such that for $t \geq -\|V\|_\infty$ we have
\begin{equation*}
    \phi_{r}\left(\frac{t+i}{t-i}\right) = e^{-\frac{1}{2} {\varepsilon} (r-1)t}.
\end{equation*}
Since $\phi_r$ is smooth, the transformer $T^{U,U}_{\phi_r^{[1]}}$ is bounded from ${\mathcal{L}}_{1}$ to ${\mathcal{L}}_1$ (see \eqref{naive_bound}), and one can compute that the $L_2(\mathbb{T})$ norms 
of $\phi_r'$ and $\phi_r''$ are bounded above by a constant multiple of $r-1$ and $(r-1)^2$ respectively, so in particular:
\begin{equation}\label{transformer_L_1_bound}
    \limsup_{r\downarrow d} \|T^{U,U}_{\phi_r^{[1]}}\|_{{\mathcal{L}}_{1}\to {\mathcal{L}}_1} < \infty.
\end{equation}
Similarly, \eqref{interpolated_naive_bound} yields:
\begin{equation}\label{transformer_L_d_1_bound}
    \limsup_{r\downarrow d} \|T^{U,U}_{\phi_{r}^{[1]}}\|_{{\mathcal{L}}_{\frac{d}{d-1},1}\to {\mathcal{L}}_{\frac{d}{d-1},1}} < \infty.
\end{equation}
Using the semigroup property of $e^{-{\varepsilon} (r-1)H}$ and the Leibniz rule, we have:
\begin{equation*}
    [M^{-r}_{\anglb{x}},e^{-{\varepsilon} (r-1) H}] = e^{-\frac{1}{2} {\varepsilon} (r-1) H}[M^{-r}_{\anglb{x}},e^{-\frac{1}{2} {\varepsilon} (r-1)H}]+[M^{-r}_{\anglb{x}},e^{-\frac{1}{2} {\varepsilon} (r-1)H}]e^{-\frac{1}{2} {\varepsilon} (r-1)H}.
\end{equation*}
Since $e^{-\frac{1}{2} {\varepsilon} (r-1)H} = \phi_r(U)$, from \eqref{DK_formula}, we have:
\begin{equation*}
    [M^{-r}_{\anglb{x}},e^{-\frac{1}{2} {\varepsilon} (r-1)H}] = T^{U,U}_{\phi_r^{[1]}}([M^{-r}_{\anglb{x}},U]).
\end{equation*}
Combining the preceding two displays, from \eqref{compatibility} it follows that
\begin{equation*}
    [M^{-r}_{\anglb{x}},e^{-{\varepsilon} (r-1)H}] = T^{U,U}_{\phi_r^{[1]}}(e^{-\frac{1}{2} {\varepsilon} (r-1)H}[M^{-r}_{\anglb{x}},U] + [M^{-r}_{\anglb{x}},U]e^{-\frac{1}{2} {\varepsilon} (r-1)H}).
\end{equation*}
Now \eqref{transformer_L_1_bound} implies that there is a constant $C_{d,{\varepsilon}}$ such that:
\begin{equation*}
    \|[M^{-r}_{\anglb{x}},e^{-{\varepsilon} (r-1)H}]\|_{1} \leq C_{d,{\varepsilon}}(\|e^{-\frac{1}{2} {\varepsilon} (r-1)H}[M^{-r}_{\anglb{x}},U]\|_1+\|[M^{-r}_{\anglb{x}},U]e^{-\frac{1}{2} {\varepsilon} (r-1)H}\|_1).
\end{equation*}
An identical argument yields:
\begin{equation*}
    \|[M^{1-r}_{\anglb{x}},e^{-{\varepsilon} (r-1)H}]\|_{\frac{d}{d-1},1} \leq C_{d,{\varepsilon}}(\|e^{-\frac{1}{2} {\varepsilon} (r-1)H}[M^{1-r}_{\anglb{x}},U]\|_{\frac{d}{d-1},1}+\|e^{-\frac{1}{2} {\varepsilon} (r-1)H}[M^{1-r}_{\anglb{x}},U]\|_{\frac{d}{d-1},1})
\end{equation*}
for a possibly larger constant $C_{d,{\varepsilon}}$.
We now concentrate on determining the ${\mathcal{L}}_{1}$ and ${\mathcal{L}}_{\frac{d}{d-1},1}$ norms of $e^{-\frac{1}{2} {\varepsilon} (r-1)H}[M^{-r}_{\anglb{x}},U]$ and $e^{-\frac{1}{2} {\varepsilon} (r-1)H}[M^{1-r}_{\anglb{x}},U]$
respectively. Identical arguments will control the ${\mathcal{L}}_{1}$ and ${\mathcal{L}}_{\frac{d}{d-1},1}$ norms of $[M^{-r}_{\anglb{x}},U]e^{-\frac{1}{2} {\varepsilon} (r-1)H}$ and $[M^{1-r}_{\anglb{x}},U]e^{-\frac{1}{2} {\varepsilon} (r-1)H}$ respectively.

Using \eqref{favourite_identity}  (once again justified by the fact that $M^{-r}_{\anglb{x}}$ preserves the domain of $H$), we have the following computations for the commutator $[M^{-r}_{\anglb{x}},U]$:
\begin{align*}
    [M^{-r}_{\anglb{x}},\frac{H+i}{H-i}] &= [M^{-r}_{\anglb{x}},1+2i(H-i)^{-1}]\\
                                            &= 2i[M^{-r}_{\anglb{x}},(H-i)^{-1}]\\
                                            &= -2i(H-i)^{-1}[M^{-r}_{\anglb{x}},H](H-i)^{-1}\\
                                            &= -2i(H-i)^{-1}[M^{-r}_{\anglb{x}},H_0](H-1)^{-1}.
\end{align*}
In view of \eqref{leibniz} and the computations leading to \eqref{F: aaa(2)}, we have:
\begin{align*}
        &(H-i)^{-1}[M^{-r}_{\anglb{x}},H_0](H-i)^{-1}\\
        &=  2r\sum_{m=1}^d(H-i)^{-1}M_{x_m\anglb{x}^{-r-2}}\partial_m(H-i)^{-1} +r(r+2-d)(H-i)^{-1}M_{\anglb{x}^{-r-2}}(H-i)^{-1}.
\end{align*}
It follows now from the triangle inequality that:
\begin{align*}
    \|e^{-\frac{1}{2} {\varepsilon} (r-1)H}[M^{-r}_{\anglb{x}},U]\|_1 &\leq C_{d,{\varepsilon}}(r\sum_{m=1}^d \|(H-i)^{-1}e^{-\frac{1}{2} {\varepsilon} (r-1)H}M^{-r-1}_{\anglb{x}}M_{x_m\anglb{x}^{-1}}\partial_m(H-i)^{-1}\|_1\\
                                                                &\quad+ r(r+d+2)\|(H-i)^{-1}e^{-\frac{1}{2} {\varepsilon} (r-1)H}M_{\anglb{x}^{-r-2}}(H-i)^{-1}\|_1).
\end{align*}
Using the fact that $x_m\anglb{x}^{-1}$, $\partial_m(H-i)^{-1}$ and $(H-i)^{-1}$ are bounded, we arrive at the bound:
\begin{equation}\label{mid_L_1_bound}
    \|e^{-\frac{1}{2} {\varepsilon} (r-1)H}[M^{-r}_{\anglb{x}},U]\|_1 \leq C_{d,{\varepsilon}}r^2\|e^{-\frac{1}{2} {\varepsilon} (r-1)H}M^{-r-1}_{\anglb{x}}\|_1.
\end{equation}
Here, once again the size of the constant may have increased.
An identical argument, replacing the ${\mathcal{L}}_{1}$ norm by the ${\mathcal{L}}_{\frac{d}{d-1},1}$ and using \eqref{transformer_L_d_1_bound} leads to the bound:
\begin{equation}\label{mid_L_d_1_bound}
    \|e^{-\frac{1}{2} {\varepsilon} (r-1)H}[M^{1-r}_{\anglb{x}},U]\|_{\frac{d}{d-1},1} \leq C_{d,{\varepsilon}}r^2\|e^{-\frac{1}{2} {\varepsilon} (r-1)H}M^{-r}_{\anglb{x}}\|_{\frac{d}{d-1},1}.
\end{equation}
Since $r > d$, \eqref{mid_L_1_bound} yields
\begin{equation*}
    \|e^{-\frac{1}{2} {\varepsilon} (r-1)H}[M^{-r}_{\anglb{x}},U]\|_1 \leq C_{d,{\varepsilon}}r^2\|e^{-\frac{{\varepsilon}(r-d)}{2}H}\|_{\infty}\|M^{d-r}_{\anglb{x}}\|_\infty\|e^{-\frac{1}{2} {\varepsilon} (d-1)H}M^{-d-1}_{\anglb{x}}\|_{1}.
\end{equation*}
and \eqref{mid_L_d_1_bound} yields
\begin{equation*}
    \|e^{-\frac{1}{2} {\varepsilon} (r-1)H}[M^{1-r}_{\anglb{x}},U]\|_{\frac{d}{d-1},1} \leq C_{d,{\varepsilon}}r^2\|e^{-\frac{{\varepsilon}(r-d)}{2}H}\|_{\infty}\|M^{d-r}_{\anglb{x}}\|_\infty\|e^{-\frac{1}{2} {\varepsilon} (d-1)H}M^{-d}_{\anglb{x}}\|_{\frac{d}{d-1},1}.
\end{equation*}
Using the fact that $(H+i)^{N}e^{-\frac{1}{2} {\varepsilon} (d-1)H}$ is bounded for any $N\geq 0$, we have:
\begin{equation*}
    \|e^{-\frac{1}{2} {\varepsilon} (d-1)H}M^{-d-1}_{\anglb{x}}\|_1 \leq C_{d,{\varepsilon}}\|(H+i)^{-(d+1)}M^{-(d+1)}_{\anglb{x}}\|_1.
\end{equation*}   
for a potentially different constant $C_{d,{\varepsilon}}$.
Theorem~\ref{main_cwikel_estimate} now provides the desired bounds, since ${\mathcal{L}}_{\frac{d}{d+1},\infty}\subset {\mathcal{L}}_1$.
\noindent Similarly,
\begin{equation*}
    \|e^{-\frac{1}{2} {\varepsilon} (d-1)H}M^{-d}_{\anglb{x}}\|_{\frac{d}{d-1},1} \leq C_{d,{\varepsilon}}\|(H+i)^{-d}M^{-d}_{\anglb{x}}\|_{\frac{d}{d-1},1}.
\end{equation*}    
Since ${\mathcal{L}}_{1,\infty}\subset {\mathcal{L}}_{\frac{d}{d-1},1}$, Theorem~\ref{main_cwikel_estimate} again yields the desired bound.

\end{proof}

\section{A residue formula}
We now proceed to the proof of the claim that:
\begin{equation*}
    {\rm Tr}_{\omega}(e^{-sH} M^{-d}_{\anglb{x}}) = \lim_{r\downarrow 1}(r-1){\rm Tr}(e^{-sH} M^{-dr}_{\anglb{x}} ).
\end{equation*}
That the left hand side makes sense is ensured by Theorem~\ref{main_cwikel_estimate}; indeed, it implies that $(H+i)^{-d}M^{-d}_{\anglb{x}} \in {\mathcal{L}}_{1,\infty}$, and since the operator $e^{-sH}(H+i)^d$ has bounded extension
it follows that $e^{-sH}M^{-d}_{\anglb{x}} \in {\mathcal{L}}_{1,\infty}$. That the right hand side makes sense will be a consequence of the arguments in this section.
The proofs of this section are achieved with some recently developed techniques in operator integration, developed originally in \cite{CSZ} and later extended in \cite[Section 5.2]{SZ_asterisque}.    

\subsection{Abstract result}
The following is an abstract operator theoretic result. It is similar to \cite[Theorem B.1]{CLMSZ}, although the assumptions are slightly different and the result is stated with an explicit norm bound.
\begin{thm}\label{main abstract thm} Let $d\geq 2$, $r>1$ let $p > \frac{d}{2}\geq 1$ and select $q \in (1,\infty)$ such that $\frac{1}{p}+\frac{1}{q} \leq 1$. Let $A,B\in B(H)$ be positive operators satisfying the following four conditions:
\begin{enumerate}[{\rm (i)}]
\item $[BA^{\frac{1}{2}},A^{\frac{1}{2}}]\in\mathcal{L}_{p,\infty};$
\item $B^{r-1}A^{r-1}\in\mathcal{L}_{q,1};$
\item $B^{r-1}[B,A^{r-1}]A\in\mathcal{L}_1;$ 
\item $(A^{1/2}BA^{1/2})^{r-1}\in\mathcal{L}_{q,1};$ 
\end{enumerate}
Then $B^rA^r-(A^{1/2}BA^{1/2})^r\in {\mathcal{L}}_1$ and
for some constants $c_{r,p,d}>0$, we have
\begin{align*}
\|B^rA^r-(A^{\frac{1}{2}}B&A^{\frac{1}{2}})^r\|_1 \leq c_{r,p,d}\Big(\|[BA^{\frac{1}{2}},A^{\frac{1}{2}}]\|_{p,\infty}\|(A^{1/2}BA^{1/2})^{r-1}\|_{q,1}\\
&\quad +\|B^{r-1}A^{r-1}\|_{q,1}\|[BA^{\frac{1}{2}},A^{\frac{1}{2}}]\|_{p,\infty}+\|B^{r-1}[BA^{\frac{1}{2}},A^{r-1}]A^{\frac12}\|_1\Big).
\end{align*}
\end{thm}

The proof of Theorem~\ref{main abstract thm} is based on the formula given in \cite[Section 5.2]{SZ_asterisque}, stated in terms of a mapping $T_r:{\mathbb{R}}\to {\mathcal{B}}({\mathcal{H}})$ defined as follows:
\begin{defn} 
Let $A$ and $B$ be positive bounded operators, and let 
$$Y := A^{1/2}BA^{1/2}.$$
We define the mapping $T_r:\mathbb{R}\to {\mathcal{B}}({\mathcal{H}})$ by,
\begin{align*}
T_r(0) &:= B^{r-1}[BA^{\frac{1}{2}},A^{r-\frac{1}{2}}]+[BA^{\frac{1}{2}},A^{\frac{1}{2}}]Y^{r-1},\\
T_r(s) &:= B^{r-1+is}[BA^{\frac{1}{2}},A^{r-\frac{1}{2}+is}]Y^{-is}+B^{is}[BA^{\frac{1}{2}},A^{\frac{1}{2}+is}]Y^{r-1-is},\quad s \neq 0.
\end{align*}
\end{defn}
We now collect some auxiliary results for the proof of Theorem~\ref{main abstract thm}.
\begin{lemma}\label{psacta main} 
Let $p$ and $d$ be as in the statement of Theorem~\ref{main abstract thm}. Suppose that $A$ and $B$ are bounded positive operators such that $[BA^{\frac12},A^{\frac12}] \in {\mathcal{L}}_{\frac{d}{2},\infty}$. Then for all $s \in {\mathbb{R}}$,
$[BA^{\frac12},A^{\frac12+is}]\in {\mathcal{L}}_{p,\infty}$, and we have:
$$\|[BA^{\frac{1}{2}},A^{\frac{1}{2}+is}]\|_{p,\infty}\leq c_p(1+|s|)\|[BA^{\frac{1}{2}},A^{\frac{1}{2}}]\|_{p,\infty}.$$
\end{lemma}
\begin{proof} 

The main result in \cite{PS_acta} asserts that for every Lipschitz continuous function $f$ on ${\mathbb{R}}$ and every $1<p' < \infty$, we have a constant $C_{p'}$ such that:
$$\|[X,f(Y)]\|_{p'}\leq C_{p'}\|f'\|_{\infty}\|[X,Y]\|_{p'},$$
whenever $Y$ is a self-adjoint operator and $X$ is a bounded operator such that $[X,Y] \in {\mathcal{L}}_{p'}$. The method of proof in \cite{PS_acta} was to construct a linear operator $T^{Y,Y}_{f^{[1]}}$ which is bounded
from ${\mathcal{L}}_{p'}$ to ${\mathcal{L}}_{p'}$ for every $1< p'< \infty$ and such that $T^{Y,Y}_{f^{[1]}}([X,Y]) = [X,f(Y)]$. 
Since $p > 1$, the ideal ${\mathcal{L}}_{p,\infty}$ is an interpolation space between ${\mathcal{L}}_{p'}$ and ${\mathcal{L}}_{p''}$ for suitable $1 < p' < p < p'' < \infty$. An interpolation argument (see e.g. \cite[p. 5225]{APPS}) further implies that $T^{Y,Y}_{f^{[1]}}$ is bounded from ${\mathcal{L}}_{p,\infty}$ to ${\mathcal{L}}_{p,\infty}$,
and we have the inequality:
$$\|[X,f(Y)]\|_{p,\infty} \leq c_{p}\|f'\|_\infty \|[X,Y]\|_{p,\infty}$$
for some constant $c_p$.

Take $X=BA^{\frac12},$ $Y=A^{\frac12}$ and $f(t)=|t|^{1+2is},$ $t\in\mathbb{R}.$ Since $\|f'\|_{\infty}\leq 1+2|s|,$ the assertion follows.
\end{proof}

\begin{lemma}\label{main abstract estimate}
Let $A,$ $B,$ $d,$ $r,$ $p$ and $q$ be as in Theorem~\ref{main abstract thm}. 
Then the map ${\mathbb{R}} \ni s \mapsto T_r(s)$ takes values in the trace class, and     
there is a constant $c_p > 0$ such that for all $s \in {\mathbb{R}}$:
\begin{align*}
    \|T_r(s)\|_1 &\leq c_p(1+|s|)\|[BA^{\frac{1}{2}},A^{\frac{1}{2}}]\|_{p,\infty}\|Y^{r-1}\|_{q,1}\\
                &\quad +c_p(1+|s|)\|B^{r-1}A^{r-1}\|_{q,1}\|[BA^{\frac{1}{2}},A^{\frac{1}{2}}]\|_{p,\infty}\\
                &\quad +\|B^{r-1}[BA^{\frac{1}{2}},A^{r-1}]A^{\frac12}\|_1.
\end{align*}
\end{lemma}
\begin{proof} 

We prove this for $s\neq 0$, the proof for $s = 0$ is identical. By the triangle inequality, we have
\begin{align*}
    \|T_r(s)\|_1 &\leq \|B^{r-1+is}[BA^{\frac{1}{2}},A^{r-\frac{1}{2}+is}]Y^{-is}\|_1+\|B^{is}[BA^{\frac{1}{2}},A^{\frac{1}{2}+is}]Y^{r-1-is}\|_1\\
                &\leq \|B^{r-1}[BA^{\frac{1}{2}},A^{r-\frac{1}{2}+is}]\|_1+\|[BA^{\frac{1}{2}},A^{\frac{1}{2}+is}]Y^{r-1}\|_1 \\
                & =: (I) + (II).
\end{align*}

Using the Leibniz rule, we have
$$[BA^{\frac{1}{2}},A^{r-\frac{1}{2}+is}] = A^{r-1}[BA^{\frac{1}{2}},A^{\frac{1}{2}+is}] + [BA^{\frac{1}{2}},A^{r-1}]A^{\frac12+is}.$$
Therefore, using \eqref{weird_holder} we get 

\begin{align*}
   (I)  & =  \|B^{r-1}[BA^{\frac{1}{2}},A^{r-\frac{1}{2}+is}]\|_1\\
      &\leq\|B^{r-1}A^{r-1}[BA^{\frac{1}{2}},A^{\frac{1}{2}+is}]\|_1+\|B^{r-1}[BA^{\frac{1}{2}},A^{r-1}]A^{\frac12+is}\|_1\\
                                                            &\leq\|B^{r-1}A^{r-1}\|_{q,1}\|[BA^{\frac{1}{2}},A^{\frac{1}{2}+is}]\|_{p,\infty}\\
                                                            &\quad +\|B^{r-1}[BA^{\frac{1}{2}},A^{r-1}]A^{\frac12}\|_1.
\end{align*}
Using Lemma~\ref{psacta main}, we have
\begin{align*}
     (I) &\leq c_p(1+|s|)\|B^{r-1}A^{r-1}\|_{q,1}\|[BA^{\frac{1}{2}},A^{\frac{1}{2}}]\|_{p,\infty}\\
                                                            &\quad+\|B^{r-1}[BA^{\frac{1}{2}},A^{r-1}]A^{\frac12}\|_1.
\end{align*}

\smallskip
On the other hand, using \eqref{weird_holder}, we have:
$$ (II) =  \|[BA^{\frac{1}{2}},A^{\frac{1}{2}+is}]Y^{r-1}\|_1\leq \|[BA^{\frac{1}{2}},A^{\frac{1}{2}+is}]\|_{p,\infty}\|Y^{r-1}\|_{q,1}.$$
Again applying Lemma~\ref{psacta main}, it follows that:
$$
 (II)  \leq c_d(1+|s|)\|[BA^{\frac{1}{2}},A^{\frac{1}{2}}]\|_{p,\infty}\|Y^{r-1}\|_{q,1}.
$$

Hence, $T_r(s) \in {\mathcal{L}}_1$ with the appropriate norm bound.

\end{proof}

By means of Lemma~\ref{main abstract estimate} and an integral formula from \cite[Theorem 5.2.1]{SZ_asterisque}, we obtain a proof of Theorem~\ref{main abstract thm}.
\begin{proof}[Proof of Theorem~\ref{main abstract thm}] 
Define the function $g_r:\mathbb{R}\to \mathbb{R}$ by
\begin{align*}
    g_r(0)  &:= 1-\frac{r}{2},\\
    g_r(2t) &:= 1-\frac{ \sinh(rt) }  { 2\sinh(t) \cosh((r-1)t) },\quad t\neq0.
\end{align*} 
The function $g_r$ is Schwartz class on ${\mathbb{R}}$ (see \cite[Remark 5.2.2]{SZ_asterisque}).
According to \cite[Theorem 5.2.1]{SZ_asterisque}, the mapping $T_r$ is continuous in the weak operator topology and
\begin{equation}\label{CSZ_integral}
    B^rA^r-Y^r = T_r(0)-\int_{\mathbb{R}} T_r(s)\widehat{g}_r(s)\,ds,
\end{equation}
where the integral on the left converges in the weak operator topology, 
and $\widehat{g}_r$ is the Fourier transform of $g_r$, scaled so that $g_r(t) = \int_{{\mathbb{R}}} \widehat{g}_r(s)e^{its}\,ds.$
Our result is based on the estimate:
\begin{equation}\label{integral_triangle_inequality}
    \|B^rA^r-Y^r\|_1 \leq \|T_r(0)\|_1+\int_{\mathbb{R}}\|T_r(s)\|_1|\widehat{g}_r(s)|ds.
\end{equation}
and a corresponding implication that if the right hand side of \eqref{integral_triangle_inequality} is finite then $B^rA^r-Y^r \in {\mathcal{L}}_1$.
This result does not immediately follow from the integral formula \eqref{CSZ_integral} since the integral only converges in the weak operator topology, however it can be justified by the argument presented in \cite[Lemma 2.3.2]{SZ_asterisque}.
Granted \eqref{integral_triangle_inequality}, we apply the bound on $\|T_r(s)\|_1$ from Lemma~\ref{main abstract estimate} to obtain:

$$\|B^rA^r-Y^r\|_1\leq RHS(0)+RHS(0)\cdot\int_{\mathbb{R}}(1+|s|)|\widehat{g}_r(s)|ds,$$
where $RHS(s)$ is the right hand side of the inequality in Lemma~\ref{main abstract estimate}.
As $g_r$ is a Schwartz class function, then so is the Fourier transform $\widehat{g}_r.$ Hence, the integral on the right hand side converges and the assertion follows.
\end{proof}

\subsection{The main residue formula}
To apply Theorem~\ref{main abstract thm} to the problem at hand, we need to verify its assumptions for the relevant operators.
\begin{lemma}\label{verification lemma} Let $A=e^{-{\varepsilon} H}$ and $B = M^{-1}_{\anglb{x}},$ where ${\varepsilon} >0$. Let $Y = A^{\frac12}BA^{\frac12}$. 
    Let $p > \frac{d}{2}$ and $q > \frac{d}{d-1}$. We have the following:
    \begin{enumerate}[{\rm (i)}]
        \item\label{vera} $[BA^{\frac{1}{2}},A^{\frac{1}{2}}]\in\mathcal{L}_{p,\infty}$ and $Y\in\mathcal{L}_{d,\infty};$
        \item\label{verb} $\|B^{r-1}A^{r-1}\|_{q,1}=O(1)$ as $r\downarrow d;$
        \item\label{verc} $\|B^{r-1}[B,A^{r-1}]A\|_1=O(1)$ as $r\downarrow d;$
        \item\label{verd} If $r > d$, we have $(A^{1/2}BA^{1/2})^{r-1} \in {\mathcal{L}}_{q,1}$. 
    \end{enumerate}
\end{lemma}
\begin{proof} 

First we verify \eqref{vera}. We begin by noting that:
\begin{equation*}
[M_{\anglb{x}}^{-1}e^{-\frac{1}{2} {\varepsilon} H},H] = [M_{\anglb{x}}^{-1},H]e^{-\frac{1}{2} {\varepsilon} H} = [M_{\anglb{x}}^{-1},H_0]e^{-\frac{1}{2} {\varepsilon} H}.
\end{equation*}
As in the proof of Lemma~\ref{first commutator lemma},  specifically \eqref{leibniz}, 
we have:
\begin{equation*}
[M^{-1}_{\anglb{x}},H_0] = 2\sum_{m=1}^d M_{\partial_m (\anglb{x}^{-1})}\partial_m+M_{H_0(\anglb{x}^{-1})}.
\end{equation*}
We may evaluate each $\partial_m(\anglb{x}^{-1})$ and $H_0(\anglb{x}^{-1})$ using \eqref{derivative_computations}, to arrive at:
\begin{equation*}
[M^{-1}_{\anglb{x}},H_0] = -2\left(\sum_{m=1}^d M_{x_m\anglb{x}^{-1}}M_{\anglb{x}}^{-2}\partial_m\right) - (3-d)M_{\anglb{x}}^{-3}.
\end{equation*} 
It follows that
\begin{align*}
[M_{\anglb{x}}^{-1}e^{-\frac{1}{2} {\varepsilon} H},H] &= [M_{\anglb{x}}^{-1},H_0]e^{-\frac{1}{2} {\varepsilon} H}\\
                                                &= -2\left(\sum_{m=1}^d M_{x_m\anglb{x}^{-1}}M_{\anglb{x}^{-2}}\partial_m\right)e^{-\frac{1}{2} {\varepsilon} H}-(3-d)M_{\anglb{x}}^{-3}e^{-\frac{1}{2} {\varepsilon} H}.
\end{align*}
The latter term is in ${\mathcal{L}}_{\frac{d}{3},\infty}$, since $(H+i)^3e^{-\frac{1}{2} {\varepsilon} H}$ is bounded
and Theorem~\ref{main_cwikel_estimate} with $p=3$ applies. For the former term, we instead use the fact that $(H_0+i)^3e^{-\frac{1}{2} {\varepsilon} H}$
is bounded, and then Theorem~\ref{main_cwikel_estimate} with $p=2$ implies that:
\begin{equation*}
[M_{\anglb{x}}^{-1}e^{-\frac{1}{2} {\varepsilon} H},H] \in {\mathcal{L}}_{\frac{d}{2},\infty} \subset {\mathcal{L}}_{p,\infty}.
\end{equation*}
Let $U = \frac{H+i}{H-i}$. It follows that
\begin{equation*}
    [M_{\anglb{x}}^{-1}e^{-\frac{1}{2} {\varepsilon} H},U] = 2i[M_{\anglb{x}}^{-1}e^{-\frac{1}{2} {\varepsilon} H},(H-i)^{-1}] = -2i(H-i)^{-1}[M_{\anglb{x}}^{-1}e^{-\frac{1}{2} {\varepsilon} H},H](H-i)^{-1} \in {\mathcal{L}}_{p,\infty}.
\end{equation*}

We now use the smooth function $\phi_{2}$ introduced in the proof of Proposition~\ref{doi_corollary} and \eqref{DK_formula} to arrive at:
\begin{equation*}
[M_{\anglb{x}}^{-1}e^{-\frac{1}{2} {\varepsilon} H},e^{-\frac{1}{2} {\varepsilon} H}] = T^{U,U}_{\phi_{2}^{[1]}}([M_{\anglb{x}}^{-1}e^{-\frac{1}{2} {\varepsilon} H},U]).
\end{equation*}
Since $p > 1$, the ideal ${\mathcal{L}}_{p,\infty}$ is an interpolation space between ${\mathcal{L}}_{1}$ and ${\mathcal{L}}_{\infty}$, and hence \eqref{naive_bound}
implies the boundedness of $T^{U,U}_{\phi_2^{[1]}}$ from ${\mathcal{L}}_{p,\infty}$ to ${\mathcal{L}}_{p,\infty}$.

Thus,
\begin{equation*}
[M_{\anglb{x}}^{-1}e^{-\frac{1}{2} {\varepsilon} H},e^{-\frac{1}{2} {\varepsilon} H}] = [BA^{\frac{1}{2}},A^{\frac{1}{2}}] \in {\mathcal{L}}_{p,\infty}. 
\end{equation*}
This yields the first part of \eqref{vera}. To see the second part (that $Y \in {\mathcal{L}}_{d,\infty}$), simply note that:
$AB=e^{-{\varepsilon} H}(H+i)\cdot (H+i)^{-1}M^{-1}_{\anglb{x}},$
so that Theorem~\ref{main_cwikel_estimate} yields $AB\in\mathcal{L}_{d,\infty}$, and:
\begin{equation*}
Y = A^{\frac12}BA^{\frac12} = BA-[BA^{\frac12},A^{\frac12}].
\end{equation*}
By taking $\frac{d}{2} < p < d$, we have already proved that $[BA^{1/2},A^{1/2}] \in {\mathcal{L}}_{d,\infty}$. Hence $Y \in {\mathcal{L}}_{d,\infty}$.

\noindent Now we prove \eqref{verb}. If $r > d$, we have:
\begin{align*}
\|B^{r-1}A^{r-1}\|_{q,1} &\leq \|B^{r-d}\|_{\infty}\|A^{r-d}\|_{\infty}\|B^{d-1}A^{d-1}\|_{q,1}.
\end{align*}
Then
$$\|B^{d-1}A^{d-1}\|_{q,1}\leq\|e^{-{\varepsilon}(d-1)H}(H+i)^{d-1}\|_{\infty}\| M^{1-d}_{\anglb{x}}  (H+i)^{1-d}\|_{q,1}.$$
Since $q > \frac{d}{d-1}$, we have:
\begin{equation*}
\|M^{1-d}_{\anglb{x}}(H+i)^{1-d}\|_{q,1} \leq \|M^{1-d}_{\anglb{x}}(H+i)^{1-d}\|_{\frac{d}{d-1},\infty}
\end{equation*}
The assertion \eqref{verb} follows now from Theorem~\ref{main_cwikel_estimate}.

\noindent Next we prove \eqref{verc}. We write $B^{r-1}[B,A^{r-1}]A$ as
$$B^{r-1}[B,A^{r-1}]A=[B^r,A^{r-1}]\cdot A-[B^{r-1},A^{r-1}]\cdot BA.$$

Thus, using \eqref{weird_holder} and our previous observation that $BA \in {\mathcal{L}}_{d,\infty}$, we have
$$\|B^{r-1}[B,A^{r-1}]A\|_1\leq\|[B^r,A^{r-1}]\|_1\|A\|_\infty+\|[B^{r-1},A^{r-1}]\|_{\frac{d}{d-1},1}\|BA\|_{d,\infty}.$$
Since $r > d$, Proposition~\ref{doi_corollary} yields the desired uniform control on $\|[B^r,A^{r-1}]\|_1$ and $\|[B^{r-1},A^{r-1}]\|_{\frac{d}{d-1},1}$ as $r\downarrow d$.

\noindent Finally, we prove \eqref{verd}. We have already proved in \eqref{vera} that $Y \in {\mathcal{L}}_{d,\infty}$. It follows that $Y^{r-1} \in {\mathcal{L}}_{\frac{d}{r-1},\infty}$. If $r > d$, then ${\mathcal{L}}_{\frac{d}{r-1},\infty}\subset {\mathcal{L}}_{\frac{d}{d-1},\infty}$,
and by our assumption on~$q$, \eqref{L_p_embeddings} implies the inclusion ${\mathcal{L}}_{\frac{d}{d-1},\infty}\subset {\mathcal{L}}_{q,1}$. This establishes \eqref{verd}.

\end{proof}
 The following lemma is essential to our results in this section, and its proof is the ultimate purpose of Lemmas \ref{main abstract thm} and \ref{verification lemma}. The importance of this result is that it allows us to replace the limit of $e^{isH}M^{-r}_{\anglb{x}}$ with the limit of $e^{-(s-{\varepsilon} d)H}Y^r$, where $Y$ is a certain compact operator. This allows the application of zeta-function residue techniques.
\begin{lemma}\label{reduction lemma} Let $s>0$ and $0 <{\varepsilon}<\frac{s}{2d}.$ Denoting $Y = e^{-\frac{\varepsilon}{2}H}M_{\anglb{x}}^{-1}e^{-\frac{\varepsilon}{2}H}$ as in Lemma~\ref{verification lemma}, we have
$$\Big\|e^{-sH}M^{-r}_{\anglb{x}}-e^{-(s-{\varepsilon} d)H}Y^r\Big\|_1=O(1),\quad r\in[d,\frac{s}{2{\varepsilon}}]$$
\end{lemma}
\begin{proof} 
Let $A$ and $B$ as in Lemma~\ref{verification lemma}. 
Splitting up $e^{-sH}$ as $e^{-(s-{\varepsilon} r)H}A^r$, we have
\begin{align*}
e^{-sH}M^{-r}_{\anglb{x}} &= e^{-(s-{\varepsilon} r)H}A^rB^r\\
                                        &= e^{-(s-{\varepsilon} r)H}\cdot (A^rB^r-Y^r) + (e^{-(s-{\varepsilon} r)H}-e^{-(s-{\varepsilon} d)H})Y^r \\
                                        & \qquad  + e^{-(s-{\varepsilon} d)H}Y^r.
\end{align*}
From the triangle inequality,
\begin{align*}
&\Big\|e^{-sH}M^{-r}_{\anglb{x}}-e^{-(s-{\varepsilon} d)H}Y^r\Big\|_1\\
&\leq\|e^{-(s-{\varepsilon} r)H}\|_{\infty}\|A^rB^r-Y^r\|_1+\|e^{-(s-{\varepsilon} r)H}-e^{-(s-{\varepsilon} d)H}\|_{\infty}\|Y^r\|_1.
\end{align*}
Consider the first summand. In view of Lemma~\ref{verification lemma}, Theorem~\ref{main abstract thm} is applicable
and, moreover, the right hand side in Theorem~\ref{main abstract thm} is bounded in $r\in[d,\frac{s}{2{\varepsilon}}].$ So, the first summand is bounded for $r\in[d,\frac{s}{2{\varepsilon}}].$ The second summand vanishes at $r=d,$ and for $r>d,$ it is estimated
via \eqref{zeta_function_inequality} as:
\begin{align*}
&\Big\|\frac{e^{-(s-{\varepsilon} r)H}-e^{-(s-{\varepsilon} d)H}}{r-d}\Big\|_{\infty}\cdot (r-d)\|Y^r\|_1\\
&\leq\sup_{t>-\|V\|_{\infty}}\Big|\frac{e^{-(s-{\varepsilon} r)t}-e^{-(s-{\varepsilon} d)t}}{r-d}\Big|\cdot (r-d)\zeta\brs{\frac rd} \|Y\|_{d,\infty}^r.
\end{align*}
Since $(r-d)\zeta(\frac{r}{d})$ is bounded in the interval $(d,\frac{s}{2{\varepsilon}}]$, the result follows.
\end{proof}

 The next theorem is the main result in this section. The proof relies on the notion of a zeta-function residue \cite[Definition 8.6.1]{LSZ2012}, a proximate notion to a Dixmier trace. If $\omega$ is an extended limit
    on $L_{\infty}(0,\infty)$, and $0 \leq A \in {\mathcal{L}}_{1,\infty}$ and $B$ is an arbitrary bounded linear operator, then the functionals $\zeta_{\omega}$ and $\zeta_{\omega,B}$
    are defined as
    \begin{equation*}
        \zeta_\omega(A) := \omega(\{t^{-1}{\rm Tr}(A^{1+t^{-1}})\}_{t > 0}),\quad \zeta_{\omega,B}(A) := \omega(\{t^{-1}{\rm Tr}(A^{1+t^{-1}}B)\}_{t > 0}).
    \end{equation*}
    It is not obvious that $\zeta_\omega$ is additive, but in fact $\zeta_\omega$
    extends to a linear functional on ${\mathcal{L}}_{1,\infty}$ and is a trace (so that $\zeta_\omega(AB) = \zeta_\omega(BA)$) \cite[Theorem 8.6.4 and Lemma 2.7.4]{LSZ2012}.
    Theorem 8.6.5 of \cite{LSZ2012} states that $\zeta_{\omega,B}(A) = \zeta_{\omega}(AB)$.
    
    The key result to which we refer is \cite[Theorem 9.3.1]{LSZ2012}, which is a zeta-function residue formula for the Dixmier trace. In particular, 
    the result implies that for a linear operator $0 \leq A \in {\mathcal{L}}_{1,\infty}$, the following are equivalent:
    \begin{enumerate}[{\rm (i)}]
        \item{} ${\rm Tr}_\omega(A) = c$ for all dilation invariant extended limits $\omega$,
        \item{} There exists a limit
                \begin{equation*}
                    c = \lim_{t\to \infty} t^{-1}{\rm Tr}(A^{1+t^{-1}}).
                \end{equation*}
    \end{enumerate}
    In particular, the theorem (combined with Theorem 8.6.5 of \cite{LSZ2012} mentioned above) entails that if $0\leq A \in {\mathcal{L}}_{1,\infty}$ and $B$ is bounded, and the limit
    \begin{equation*}
        c = \lim_{t\to\infty} t^{-1}{\rm Tr}(A^{1+t^{-1}}B)
    \end{equation*}
    exists, then ${\rm Tr}_\omega(AB) = c$ for all dilation-invariant extended limits $\omega$. We would like to apply this result to the following theorem, with $A = M^{-d}_{\anglb{x}}$
    and $B = e^{-sH}$, but the theorem does not directly apply since $M^{-1}_{\anglb{x}}$ is not even compact, let alone in ${\mathcal{L}}_{1,\infty}$. To overcome this difficulty, we apply Lemma \ref{reduction lemma}.
    The full details are as follows.

\begin{thm}\label{main_residue_formula}
 Assume that there exists the limit
\begin{equation*}
    E = d^{-1}\lim_{r\downarrow d}(r-d){\rm Tr}(e^{-sH}M^{-r}_{\anglb{x}})
\end{equation*}
then for any Dixmier trace ${\rm Tr}_\omega$ and all $s > 0$ we have
$${\rm Tr}_{\omega}(e^{-sH}M^{-d}_{\anglb{x}}) = E.$$
\end{thm}
\begin{proof}
Let $Y = e^{-\frac{\varepsilon}{2}H}M_{\anglb{x}}^{-1}e^{-\frac{\varepsilon}{2}H}$ (as in Lemmas \ref{verification lemma} and \ref{reduction lemma}).
Lemma~\ref{reduction lemma} implies 
$$
    E  = d^{-1} \lim_{r\downarrow d}(r-d){\rm Tr}(e^{-(s-{\varepsilon} d)H}Y^r).
$$
Hence for every extended limit $\omega(f) = \omega - \lim_{t \to \infty} f(t)$ 
on $L_\infty(0,\infty)$ , we have (taking $\frac{r}{d}-1=\frac1t$)
$$
E = \omega\Big(\frac1t{\rm Tr}(e^{-(s-{\varepsilon} d)H}(Y^d)^{1+\frac1t}\Big).
$$
According to \cite[Theorem 8.6.5]{LSZ2012} (with $A=Y^d$ and $B=e^{-(s-{\varepsilon} d)H}$), it follows that
$$ 
    E = \zeta_{\omega}(e^{-(s-{\varepsilon} d)H}Y^d).
$$
where $\zeta_\omega \colon {\mathcal{L}}_{1,\infty} \to {\mathbb{C}}$ is the zeta function associated with the extended limit $\omega$ (see \cite[Definition 8.6.1 and Theorem 8.6.4]{LSZ2012}).
Appealing to Lemma~\ref{reduction lemma} with $r=d$ and taking into account that $\zeta_{\omega}$ vanishes on $\mathcal{L}_1,$ we obtain
$$
    E = \zeta_{\omega}(e^{-sH}M^{-d}_{\anglb{x}}).
$$
Since $\zeta_{\omega}$ is a trace (by \cite[Theorem 8.6.4 and Lemma 2.7.4]{LSZ2012}), 
it follows that 

$$
    E = \zeta_{\omega}(e^{-\frac12sH}M^{-d}_{\anglb{x}}e^{-\frac12sH}).
$$
Combining this with \cite[Theorem 9.3.1]{LSZ2012} gives 

$$
E = {\rm Tr}_{\omega}(e^{-\frac12sH}M^{-d}_{\anglb{x}}e^{-\frac12sH}).
$$
It follows that
$
    E = {\rm Tr}_{\omega}(e^{-sH}M^{-d}_{\anglb{x}}).
$
\end{proof}

\section{Formula for the density of states}

Our next step is to show that for all $s > 0$ we have
\begin{equation*}
    \lim_{R\to\infty} \frac{1}{|B(0,R)|}{\rm Tr}(e^{-sH}M_{\chi_{B(0,R)}}) = \frac{d}{\omega_d}\lim_{r\downarrow 1} (r-1){\rm Tr}(e^{-sH}M_{\anglb{x}}^{-dr}).
\end{equation*}
For each $s>0$, the operator $e^{-sH}$ is an integral operator \cite[Corollary 25.9]{simon_functional}, denote its kernel by $K_{s,V}$.
We shall prove the equivalent statement that
\begin{equation*}
    \lim_{R\to\infty} \frac{1}{|B(0,R)|} \int_{B(0,R)} K_{s,V}(x,x)\,dx = \frac{d}{\omega_d}\lim_{r\downarrow 1} \,(r-1)\int_{{\mathbb{R}}^d} {\anglb{x}}^{-dr}K_{s,V}(x,x)\,dx.
\end{equation*}
\noindent Although the kernel $K_{s,V}$ is only \emph{a priori} defined pointwise-almost everywhere, we understand the meaning of $K_{s,V}(x,x)$ in a Lebesgue averaged sense, as 
justified by Brislawn's theorem \cite[Theorem 3.1]{Brislawn-trace-class-1988}.
The following is a routine abelian theorem.
\begin{lemma}\label{abelian_lemma}
    Let $F$ be a bounded measurable function on ${\mathbb{R}}^d$ and assume that there is $c \in {\mathbb{C}}$ such that:
    \begin{equation*}
        \int_{B(0,R)} F(t)\,dt = cR^d + o(R^{d}),\quad R\to\infty.
    \end{equation*}
    Then:
    \begin{equation*}
        \int_{{\mathbb{R}}^d} {\anglb{t}}^{-dr}F(t)\,dt = \frac{c}{r-1}+o\brs{\frac{1}{r-1}},\quad r\downarrow 1.
    \end{equation*}
    More concisely, we have:
    \begin{equation*}
        \lim_{R\to\infty} \frac{1}{|B(0,R)|}\int_{B(0,R)} F(t)\,dt = \frac{1}{|B(0,1)|}\lim_{r\downarrow 1}\,(r-1)\int_{{\mathbb{R}}^d} {\anglb{t}}^{-dr}F(t)\,dt,
    \end{equation*}
    whenever the left hand side exists.
\end{lemma}
\begin{proof}
    Write ${\anglb{t}}^{-dr}$ as an integral of an indicator function:
    \begin{align*}
        {\anglb{t}}^{-dr} = \int_0^{{\anglb{t}}^{-dr}}\,d\theta &= \int_{0}^1 \chi_{[0,{\anglb{t}}^{-dr})}(\theta)\,d\theta\\
                                                                &= \int_{0}^1 \chi_{[0,(\theta^{-\frac{2}{dr}}-1)^{1/2})}(|t|)\,d\theta.
    \end{align*}
    Thus by Fubini's theorem:
    \begin{align*}
        \int_{{\mathbb{R}}^d} {\anglb{t}}^{-dr}F(t)\,dt &= \int_0^1 \int_{{\mathbb{R}}^d} \chi_{[0,(\theta^{-\frac{2}{dr}}-1)^{1/2})}(|t|)F(t)\,dtd\theta\\
                                                &= \int_0^1 \left(\int_{B(0,(\theta^{-\frac{2}{dr}}-1)^{1/2})} F(t)\,dt\right)d\theta. 
    \end{align*}
    With the change of variable $\theta = \anglb{R}^{-dr},$ we have:
    \begin{align*}
        \int_{{\mathbb{R}}^d} {\anglb{t}}^{-dr}F(t)\,dt &= \int_0^\infty \frac{drR}{{\anglb{R}}^{dr+2}}\int_{B(0,R)}F(t)\,dt dR\\
                                                &= dr\int_{0}^\infty \frac{R^{d+1}}{{\anglb{R}}^{dr+2}}\left(\frac{1}{R^d}\int_{B(0,R)} F(t)\,dt\right)\,dR.
    \end{align*}
    Our assumption is that:
    \begin{equation*}
        \frac{1}{R^d}\int_{B(0,R)}F(t)\,dt = c + \rho(R)
    \end{equation*}
    where $\rho(R) = o(1)$ as $R\to \infty$.
    Therefore:
    \begin{equation*}
        \int_{{\mathbb{R}}^d} {\anglb{t}}^{-dr}F(t)\,dt = drc\int_0^\infty \frac{R^{d+1}}{{\anglb{R}}^{dr+2}}\,dR + dr\int_0^\infty \frac{R^{d+1}}{{\anglb{R}}^{dr+2}}\rho(R)\,dR.
    \end{equation*}
    The former integral evaluates to $\frac{1}{2}\frac{\Gamma(d(r-1)/2)\Gamma(1+d/2)}{\Gamma(1+dr/2)}$. Since the gamma function has a pole with residue $1$ at $0$\footnote{This follows from the identity $\Gamma(z) = \frac{1}{z}\Gamma(z+1)$}, we have:
    \begin{equation*}
        \int_{{\mathbb{R}}^d} {\anglb{t}}^{-dr}F(t)\,dt = \frac{c}{r-1}+O(1) + dr\int_0^\infty \frac{R^{d+1}}{{\anglb{R}}^{dr+2}}\rho(R)\,dR,\quad r\downarrow 1.
    \end{equation*}
    It remains to show that the last summand is $o\brs{\frac 1{r-1}}.$    
    
It is enough to prove that 
\begin{equation} \label{F: (*)}
 \qquad  \lim_{r\to 1^+} (r-1) \int_1^\infty \frac{R^{d+1}}{\anglb{R}^{dr+2}} \rho(R)\,dR = 0.
\end{equation}
    By assumption, $\rho(R)\to 0$ as $R\to\infty$. 
If $\rho$ has compact support, then \eqref{F: (*)} is automatically true by the dominated convergence theorem. 
Note that if $\abs{\rho(R)} \leq a,$ then:
$$
    (r-1)\int_1^\infty \frac{R^{d+1}}{\anglb{x}^{dr+1}}|\rho(R)| dR \leq a/d.
$$
Choose any ${\varepsilon}>0$ and 
write $\rho$ as 
$
   \rho = \rho_{c} + \rho_{s}
$
where $\rho_{c}$ has compact support and $\abs{\rho_{s}}$ is bounded by ${\varepsilon} d.$
Then by the triangle inequality:
$$
    |(r-1)\int_1^\infty \frac{R^{d+1}}{\anglb{R}^{dr+2}} \rho(R) dR| \leq (r-1)\int_1^\infty R^{d-rd-1} |\rho_{c}(R)| dR + {\varepsilon}.
$$
Choose $r$ sufficiently close to $1$ such that the first term is less than ${\varepsilon}ilon.$ We can do this, since \eqref{F: (*)} holds for $\rho_c.$ But then:
$$
|(r-1)\int_1^\infty R^{-r} \rho(R) dR| \leq 2{\varepsilon}
$$
and this gives the result.
    
\end{proof}

\begin{cor}\label{abelian_corollary}
    For all $s > 0$, if the density of states measure $\nu_H$ exists, then:
    \begin{equation*}
        \int_{{\mathbb{R}}} e^{-s\lambda}\,d\nu_H(\lambda) = \frac{1}{|B(0,1)|}\lim_{r\downarrow 1}\,(r-1)\int_{{\mathbb{R}}^d} {\anglb{x}}^{-dr}K_{s,V}(x,x)\,dx.
    \end{equation*}
    That is,
    \begin{equation*}
        \int_{{\mathbb{R}}} e^{-s\lambda}\,d\nu_H(\lambda) = \frac{1}{|B(0,1)|}\lim_{r\downarrow 1}\,(r-1){\rm Tr}( e^{-sH}  M^{-dr}_{\anglb{x}} ).
    \end{equation*}
    whenever $\nu_H$ exists.
\end{cor}
\begin{proof}
    By definition (c.f. \eqref{def_of_DOS}), we have:
    \begin{align*}
        \int_{{\mathbb{R}}} e^{-s\lambda}\,d\nu_H(\lambda) &= \lim_{R\to\infty} \frac{1}{|B(0,R)|}{\rm Tr}(M_{\chi_{B(0,R)}}e^{-sH})\\
                                                    &= \lim_{R\to\infty} \frac{1}{|B(0,R)|}\int_{B(0,R)} K_{s,V}(x,x)\,dx.
    \end{align*}
    We now conclude the proof by an application of Lemma~\ref{abelian_lemma}. The only condition of Lemma~\ref{abelian_lemma} which needs to be checked is that $x\mapsto K_{s,V}(x,x)$ is essentially bounded on ${\mathbb{R}}^d$. This is \cite[Corollary 25.9]{simon_functional}.
\end{proof}

\begin{remark}\label{laplace_uniqueness}
    Our proof crucially relies on the following well-known property of the Laplace transform of measures \cite[Theorem II.6.3]{Widder1941}: if $\nu$ and $\mu$ are complex Borel measures supported on some semiaxis $[-C,\infty)$
    such that:
    \begin{equation*}
        \int_{{\mathbb{R}}} e^{-st} \,d\nu(t) = \int_{{\mathbb{R}}} e^{-st}\,d\mu(t)
    \end{equation*}
    for all $s > 0$ (in particular, both integrals as Lebesgue integrals for all $s > 0$), then $\nu = \mu$.
    
    An easy way to see this is as a consequence of the Stone-Weierstrass theorem \cite[\S 5.7]{Rudin_functional}. Without loss of generality, $C = 0$ and $\mu = 0$. Then we have a measure $\nu$ on $[0,\infty)$ such
    that $\int_{0}^\infty e^{-st}\,d\nu(t) = 0$ for all $s > 0$. It follows that $\int_{0}^\infty g(t)\,d\nu(t) =0$ for all functions $g$ which are a finite linear span of functions in $\{e^{-st}\}_{s > 0}$.
    
    However the linear span of $\{e^{-st}\}_{s > 0}$ is a subalgebra of the set $C_0([0,\infty))$ which separates points, hence every $f \in C_0([0,\infty))$ is a uniform limit
    of functions in the linear span of $\{e^{-st}\}_{s > 0}$. Let $f \in C_c([0,\infty))$ be a continuous compactly supported function, and select a sequence $\{g_n\}_{n\geq 0}$ of functions in the linear span of $\{e^{-st}\}_{s>0}$
    which uniformly approximate the continuous compactly supported function $t\mapsto f(t)e^t$ as $n\to\infty$.
    Thus we have:
    \begin{equation*}
        \left|\int_0^\infty f\,d\nu-\int_0^\infty g_n(t)e^{-t}\,d\nu(t)\right|\leq \sup_{t \geq 0} |e^tf(t)-g_n(t)|\int_0^\infty e^{-t}\,d|\nu|(t).
    \end{equation*}
    Since each $\int_{0}^\infty g_n(t)e^{-t}\,d\nu$ vanishes and also $\int_{0}^\infty e^{-t}\,d|\nu|(t) < \infty$ by the assumption that each $e^{-st}$ is $\nu$-integrable in the Lebesgue sense, it follows that $\int_0^\infty f(t)\,d\nu = 0$ for all continuous compactly supported continuous functions $f$. Hence by the Riesz theorem, $\nu = 0$.
\end{remark}

\begin{proof}[Proof of Theorem~\ref{T: main}]
    Corollary~\ref{abelian_corollary} yields:
    \begin{equation*}
        \int_{{\mathbb{R}}} e^{-s\lambda}\,d\nu_H(\lambda) = \frac{1}{|B(0,1)|}\lim_{r\downarrow 1}\,(r-1){\rm Tr}( e^{-sH}  M^{-dr}_{\anglb{x}} ).
    \end{equation*}
    Theorem~\ref{main_residue_formula} identifies the limit above as being exactly:
    \begin{align*}
        \lim_{r\downarrow 1}\,(r-1){\rm Tr}( e^{-sH}  M^{-dr}_{\anglb{x}} ) &= \frac{1}{d}\lim_{r\downarrow d} (r-d){\rm Tr}(e^{-sH}M^{-r}_{\anglb{x}})\\
                                                                        &= {\rm Tr}_\omega(e^{-sH}M^{-d}_{\anglb{x}}).
    \end{align*}
    Therefore,
    \begin{equation*}
        \int_{{\mathbb{R}}} e^{-s\lambda}\,d\nu_H(\lambda) = \frac{1}{|B(0,1)|}{\rm Tr}_\omega(e^{-sH}M^{-d}_{\anglb{x}}).
    \end{equation*}
    Theorem~\ref{main_cwikel_estimate} implies that if $f$ is a Borel function on $\mathbb{R}$ such that $t\mapsto |f(t)|\anglb{t}^d$ is bounded, then:
    \begin{equation}\label{measure_growth_bound}
        |{\rm Tr}_\omega(f(H)M^{-d}_{\anglb{x}})| \leq C\sup_{t \in {\mathbb{R}}} |f(t)|\anglb{t}^d
    \end{equation}
    for some constant $C$. From the Riesz theorem, it follows that the functional $f\mapsto {\rm Tr}_\omega(f(H)M^{-d}_{\anglb{x}})$ is represented by a Borel measure $\mu$ on $\mathbb{R}$,
    \begin{equation*}
        {\rm Tr}_{\omega}(f(H)M^{-d}_{\anglb{x}}) = \int_{\mathbb{R}} f\,d\mu,\quad f \in C_c(\mathbb{R}).
    \end{equation*}
    This identity is only \emph{a priori} valid for continuous compactly supported functions, but we may include the function $f(t) = e^{-st}$, for $s > 0$, as follows. 
    Select a sequence $\{f_n\}_{n=0}^\infty \subset C_c({\mathbb{R}})$ such that as $n\to\infty$ we have:
    $$\sup_{t > -\|V\|_{\infty}} |e^{-st}-f_n(t)|\anglb{t}^d\rightarrow 0.$$
    It follows that $\int_{\mathbb{R}} f_n\,d\mu\rightarrow \int_{\mathbb{R}} e^{-st}\,d\mu(t)$ and \eqref{measure_growth_bound} implies that ${\rm Tr}_\omega(f_n(H)M^{-d}_{\anglb{x}})\rightarrow {\rm Tr}_\omega(e^{-sH}M^{-d}_{\anglb{x}})$.
    Hence,
    \begin{equation*}
        \int_{\mathbb{R}} e^{-st}\,d\mu(t) = {\rm Tr}_\omega(e^{-sH}M^{-d}_{\anglb{x}}) = |B(0,1)|\int_{\mathbb{R}} e^{-st}\,d\nu_H(t),\quad s > 0.
    \end{equation*}
    Uniqueness for the Laplace transform (Remark~\ref{laplace_uniqueness}) gives the equality of measures, $\mu = |B(0,1)|\nu_H = \frac{\omega_d}{d}\nu_H$, and this is the desired equality.
\end{proof}

\section{Acknowledgements}
The authors wish to thank the anonymous referee for helpful comments and suggestions. F. S. is partially supported by the Australian Research Council grant FL170100052.

\end{document}